\documentclass[journal]{IEEEtran}
\usepackage{lineno}
\usepackage{cite}
\usepackage{hyperref}
\usepackage[tbtags]{amsmath}
\usepackage{amssymb,amsfonts}
\usepackage{amsthm}
\usepackage{amsbsy}

\usepackage{graphicx}
\usepackage{textcomp}
\usepackage{xcolor}
\usepackage{float}
\usepackage{mathrsfs}
\usepackage{mathtools}
\usepackage{multirow}
\usepackage{algpseudocode}
\usepackage{comment}
\usepackage[inline, shortlabels]{enumitem}
\makeatletter
\newif\if@restonecol

\usepackage{algorithm}

\usepackage{caption}
\usepackage{subcaption}
\usepackage{ragged2e}

\usepackage{dsfont}
\usepackage{bbm}
\usepackage{stfloats}
\usepackage{bm}
\usepackage{orcidlink}

\newtheorem{corollary}{Corollary}

\theoremstyle{definition}
\newtheorem{theorem}{Theorem}

\newtheorem{lemma}{Lemma}

\newtheorem{proposition}{Proposition}

\newcommand{\biggg}{\bBigg@{3}}
\newcommand{\Biggg}{\bBigg@{3.5}}
\makeatother
\begin{document}

\title{ Rotatable Antenna-Enabled Near-Field Integrated Sensing and Communication}

\author{Zequan~Wang, Liang~Yin, Yitong~Liu, and Hongwen~Yang

\thanks{Zequan~Wang, Liang~Yin, Yitong~Liu, and Hongwen~Yang are with the School of Information and Communication Engineering, Beijing University of Posts and Telecommunications, Beijing, 100876, China (e-mail: \{zequanwang, YinL, liuyitong, yanghong\}@bupt.edu.cn).}

\thanks{(Corresponding author: Liang~Yin)}
}

\maketitle

\begin{abstract}
In this paper, we propose leveraging rotatable antennas (RAs) to enhance
near-field communication and sensing by exploiting a new
orientation-domain spatial degree-of-freedom (DoF) provided by element-wise
antenna rotation. Specifically, we investigate an RA-enabled near-field
integrated sensing and communication (ISAC) system with sub-connected
hybrid beamforming, where each transmit RA can
independently adjust its boresight direction under a practical rotation
constraint. A spherical-wave channel model incorporating
orientation-dependent antenna gains is established to characterize
multi-user communication and target sensing in the presence of clutters.
Based on this model, a weighted communication-sensing utility
maximization problem is formulated by jointly optimizing the receive
beamformer, digital beamformer, analog beamformer, and RA
boresight directions. To solve the resulting non-convex problem, an
alternating optimization algorithm is developed by combining fractional
programming, Riemannian optimization, and a spherical-cap
Frank--Wolfe-based boresight update. To further understand the impact of
RA rotation on near-field sensing, we derive a closed-form root Cramér--Rao bound
(RCRB) expression. Simulation results demonstrate the convergence and
effectiveness of the proposed algorithm. It is shown that the RA-enabled hybrid design can match or even outperform the
fully-digital FPA benchmark in some regimes, indicating that the
orientation-domain DoF introduced by element-wise rotation can compensate
for limited RF chains. The RCRB and beampattern results
further show that RA rotation improves off-broadside sensing accuracy,
enhances range-domain focusing, and suppresses same-angle clutters in the
near field.

\end{abstract}


\section{Introduction}
Integrated sensing and communication (ISAC) has been widely recognized as a key enabling technology for sixth-generation (6G) wireless networks, as it allows communication and radar sensing functionalities to share spectrum resources, hardware platforms, and transmitted waveforms \cite{Liu2018},\cite{Liu2022},\cite{Liu2020}. ISAC has also been identified by the International Telecommunication Union (ITU) as one of the representative usage scenarios for future 6G systems \cite{ITU}. 
To support high-rate transmission and fine sensing resolution, future ISAC systems are expected to exploit millimeter-wave/terahertz bands and extremely large-scale multiple-input multiple-output (XL-MIMO) arrays. The resulting short wavelengths and large apertures significantly enlarge the Rayleigh distance, making near-field effects non-negligible in practical base-station (BS) deployments \cite{Lu2022,Cui2022}.
Unlike the far-field plane-wave model, near-field spherical-wave propagation yields range-dependent array responses. Hence, beamforming can be performed over both angular and radial dimensions, providing range-resolved spatial selectivity for localized target illumination and the discrimination of objects with similar angular directions but different ranges in ISAC \cite{Wang2023}.

Despite the range-domain selectivity of near-field propagation, near-field ISAC remains challenging due to the tightly coupled resource sharing and interference control. A dual-functional BS reuses the same power budget, antenna aperture, RF hardware, and waveform resources for multi-user downlink transmission and target sensing.
Fully-digital beamforming provides high spatial flexibility, but assigning one RF chain to each antenna element is impractical for XL-MIMO arrays due to the excessive cost, power consumption, and calibration overhead. Hybrid beamforming has therefore been widely adopted to reduce the number of RF chains \cite{Yu2016,S2019}. In particular, sub-connected architectures, where each RF chain drives a disjoint antenna subarray, require much fewer phase shifters than fully-connected networks and are attractive for large-scale implementation \cite{Du2018,WangX2022,Zhu2024}.
However, the resulting block-diagonal constant-modulus analog precoder and the limited RF-chain dimension restrict the available spatial degrees-of-freedom (DoFs). For conventional fixed-position antenna (FPA) arrays, whose element positions and radiation boresights are predetermined, such DoF loss is difficult to compensate without increasing the number of antennas or RF chains.

To overcome the limited spatial adaptability of FPA arrays, reconfigurable antenna technologies have recently attracted increasing attention. Movable antennas (MAs) \cite{ZhuL20241,ZhuL20242,Ding2025,Lyu2025} and fluid antenna systems (FASs) \cite{D2026} reshape wireless channels through position-domain reconfigurability, and have shown performance gains in communication and sensing. Six-dimensional movable antenna (6DMA) systems further combine position and orientation reconfiguration, enabling the array response and radiation characteristics to adapt to different propagation environments \cite{Shao2026,Shao2025,Shao20251,Shao20252,Shao20253}. However, these gains usually require translational movement, additional deployment space, and non-negligible mechanical control overhead \cite{{Li2024}}. 

As a simplified realization of 6DMAs, rotatable antennas (RAs) retain orientation reconfigurability while avoiding translational movement \cite{Zheng2026}. Specifically, the three-dimensional (3D) boresight of each antenna can be independently adjusted via mechanical or electronic means, requiring only local rotation and thus facilitating compact-array implementation and compatibility with existing RF front-ends. By aligning directional radiation patterns with desired propagation directions and reducing the gain toward undesired ones, RAs provide an additional orientation-domain DoF. 
Depending on the rotation granularity, RA designs can be broadly classified into array-wise and element-wise configurations \cite{Zheng2026}. 
In array-wise RA designs, the antenna array follows a common orientation configuration, which is simple to implement but provides limited orientation flexibility. In contrast, element-wise RA designs allow individual antennas to adjust their boresights independently, offering richer orientation-domain DoFs for reshaping the effective spatial response.
For element-wise systems, \cite{Zheng20261} established RA channel models and boresight optimization methods for wireless communications, \cite{Xiong2025} developed efficient channel estimation schemes for RA-enabled systems, and \cite{Wang20261} proposed a two-level RA architecture for sensing-aided secure multicast. These studies have shown the potential of antenna rotation for spatial channel reconfiguration.

Following these advances, RA techniques have recently been extended to ISAC scenarios. For example, \cite{Zhou2025} jointly optimized transmit beamforming and the array rotation angle to balance the communication sum rate and the sensing CRB. \cite{Wang2026} incorporated array rotation into sub-connected hybrid beamforming for far-field ISAC under limited RF chains. \cite{Zhang20262} further investigated RA-aided near-field ISAC and analyzed the corresponding range and angle estimation CRBs.
However, existing RA-enabled ISAC designs mainly rotate the entire array as a whole, where all antenna elements share a common orientation configuration. This array-wise model leaves element-wise orientation diversity unexplored. In the near field, different antenna elements generally observe a given spatial point along distinct propagation directions, making per-element boresight control fundamentally richer than array-wise rotation. Therefore, jointly exploiting element-wise RA boresight control for near-field ISAC remains to be addressed.

Based on the above discussion, this paper investigates an element-wise RA-enabled near-field ISAC system. To reduce hardware cost and implementation overhead, we consider a sub-connected hybrid beamforming architecture at the transmitter, where the transmit array is composed of independently rotatable antennas, while the receive array employs conventional FPAs for echo reception. Under this practical architecture, the element-wise boresight directions of the transmit RAs are jointly designed with the hybrid beamformers to improve the overall communication-sensing performance in cluttered near-field environments. Moreover, to reveal the impact of element-wise antenna rotation on target sensing accuracy, we derive the corresponding RCRB expression for target location estimation. 
To the best of the authors' knowledge, element-wise RA-enabled near-field ISAC with sub-connected hybrid beamforming has not yet been investigated. The main contributions of this article are summarized as follows.

\begin{itemize}
    \item We propose an element-wise RA-enabled near-field ISAC framework with sub-connected hybrid beamforming. 
    To better reflect practical sensing scenarios, clutter scatterers are explicitly modeled as sensing interference to the target echo.
    A near-field channel model is developed by combining spherical-wave propagation with orientation-dependent RA gains. Based on this model, we formulate a weighted communication-sensing utility maximization problem under the sub-connected hybrid beamforming structure.

    \item We develop an efficient alternating optimization (AO) algorithm to solve the formulated non-convex problem. Fractional programming (FP) is adopted to handle the coupled SINR and SCNR terms. The receive beamformer admits a closed-form solution, while the digital beamformer is obtained via the KKT conditions and bisection search. For the analog beamformer, the sub-connected structure is reformulated into a compact unit-modulus vector and optimized over the complex circle manifold. For the RA boresight directions, we propose a spherical-cap Frank--Wolfe-based update with a closed-form linear oracle, which directly accounts for the practical per-antenna rotation constraint.
    
    \item To further analyze the impact of RA rotation on target sensing, we derive a root Cram\'er--Rao bound (RCRB) expression for near-field target location estimation. By treating the unknown complex target response as a nuisance parameter and eliminating it through the Schur complement, the derived RCRB captures both near-field steering derivatives and orientation-dependent RA gain derivatives, thereby revealing how RA boresight control affects range and angle estimation.
    
    \item Numerical simulations are provided to evaluate the effectiveness of the proposed algorithm and the advantages of the sub-connected RA-enabled ISAC system. The element-wise RA-enabled hybrid scheme outperforms conventional FPA and fixed-RA schemes, and can approach or even surpass the fully-digital FPA benchmark in certain regimes. The RCRB and beampattern results further demonstrate that RA rotation improves off-broadside sensing accuracy, enhances range-domain focusing, and suppresses same-angle clutters in the near field.
\end{itemize}

Notation: $a/A$, $\mathbf{a}$, $\mathbf{A}$, and $\mathcal{A}$ denote a scalar, a vector, a matrix, and a set, respectively. $(\cdot)^{\mathrm T}$, $(\cdot)^{\mathrm H}$, $\odot$, $\|\cdot\|_2$, $|\cdot|$, $\|\cdot\|_F$, $\operatorname{Tr}\{\cdot\}$ denote the transpose, conjugate transpose, Hadamard product, Euclidean norm, absolute value, Frobenius matrix norm and trace operations, respectively. $j=\sqrt{-1}$ represents the imaginary unit. $\mathbb{C}^{M \times N}$ and $\mathbb{R}^{M \times N}$ are the sets for complex and real matrices of $M \times N$ dimensions, respectively. $\mathbf{I}_N$ is the identity matrix of order $N$. Finally, $[\cdot]_{(m,n)}$ denotes the $(m,n)$-th element of a matrix.
\section{System Model}
\subsection{RA-BS Model}
\begin{figure}[!t]
    \centering
    \includegraphics[height=0.40\textwidth]{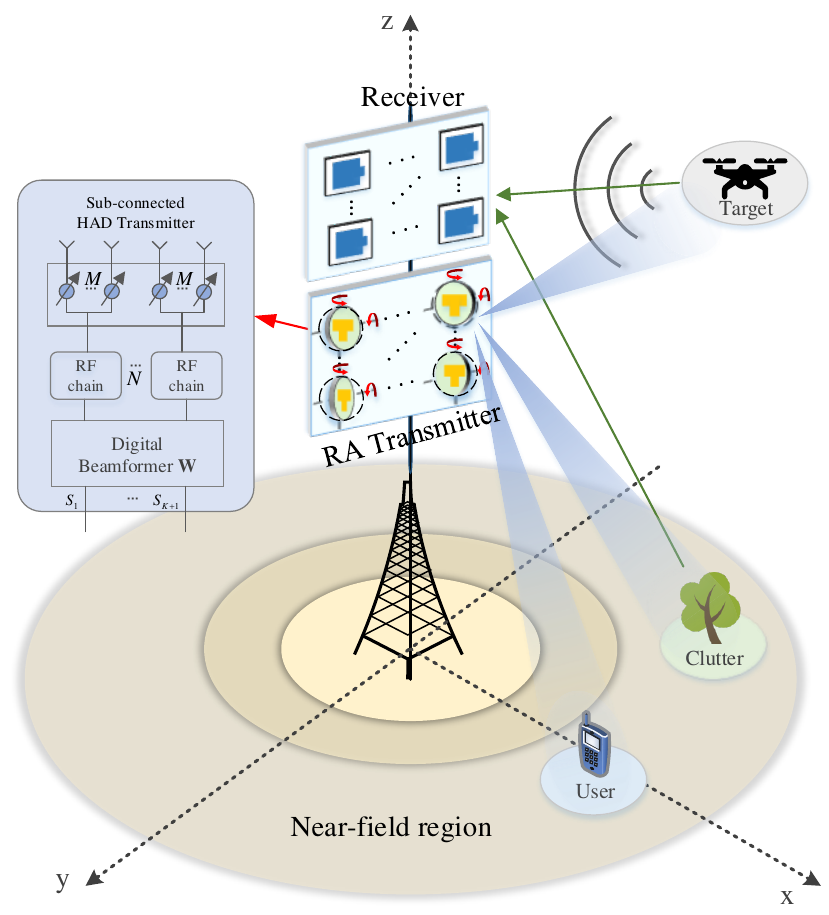}
    \captionsetup{font={small}}
    \caption{\justifying Illustration of the proposed RA-aided near-field ISAC system.}
    \label{scene0827}
\end{figure}
\begin{figure}[!t]
    
    \centering
    \begin{subfigure}{0.48\linewidth}
        \centering
        \includegraphics[width=\linewidth]{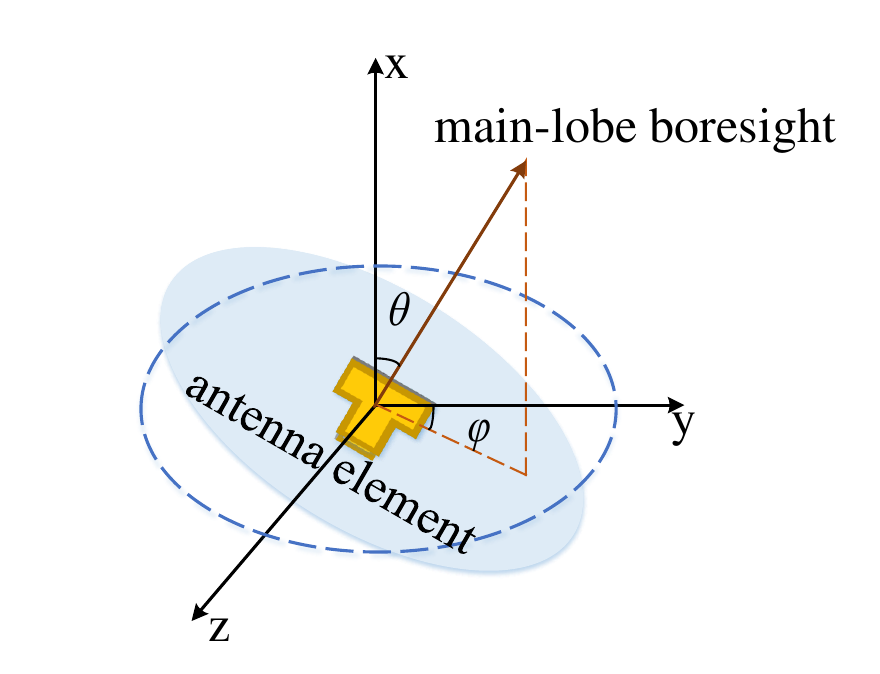}
        \caption{\justifying Rotational angles.}
        \label{b}
    \end{subfigure}
    \hfill
    \begin{subfigure}{0.48\linewidth}
        \centering
        \includegraphics[width=\linewidth]{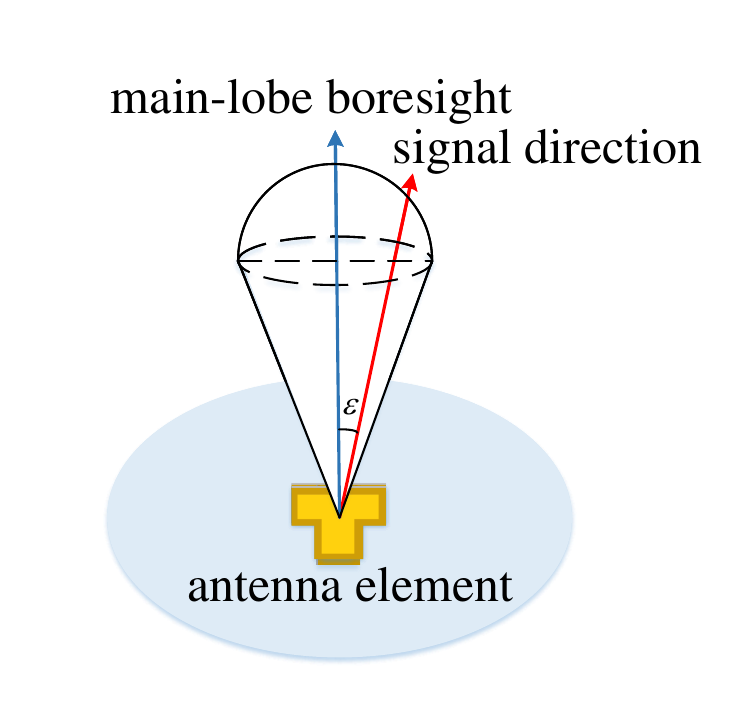}
        \caption{\justifying directional power pattern.}
        \label{a}
    \end{subfigure}
    \captionsetup{font={small}}
    \caption{\justifying Directional power pattern and rotation-angle of RA.}
    \label{fig:shiyitu}
\end{figure}

As shown in Fig.~\ref{scene0827}, we consider an RA-enabled near-field ISAC system. The system comprises a dual-functional radar-communication BS that simultaneously serves $K$ single-antenna users and senses one target in the presence of $C$ clutters, which act as sensing interference. The BS is equipped with a transmit UPA (TP) and a receive UPA (RP), both placed on the $y$-$z$ plane and oriented toward the positive $x$-axis. The TP employs RAs for directional transmission, while the RP uses FPAs for echo reception.
The TP and RP consist of \(N_t=N_y^{\rm t}N_z^{\rm t}\) and \(N_r\) antennas, respectively, with inter-element spacing \(\delta\). Taking the TP center as the origin of its local coordinate system, the position of the \(n\)-th transmit RA is denoted by
\begin{equation}
    \mathbf t_n^{\rm t}=
    \big[0,\; y_n^{\rm t},\; z_n^{\rm t}\big]^{\rm T},
    \quad n=1,\ldots,N_t.
\end{equation}

The TP employs a sub-connected hybrid beamforming architecture with $B$ RF chains. The transmit array is partitioned into $B$ non-overlapping RF-connected subarrays, each driven by one dedicated RF chain. Let $M=N_t/B$ denote the number of antennas connected to each RF chain, where $M$ is assumed to be an integer. The antenna index set connected to the $b$-th RF chain is given by
\begin{equation}
    \mathcal R_b
    =
    \{(b-1)M+1,\ldots,bM\},
    \quad b=1,\ldots,B .
\end{equation}

Each transmit RA can independently adjust its boresight direction. As depicted in Fig.~\eqref{b}, the boresight direction of the $n$-th RA is denoted by $\mathbf p_{n_t}$ and parameterized as
\begin{equation}
    \mathbf p_{n}
    =
    \left[
    \cos\theta_{n},\;
    \sin\theta_{n}\cos\varphi_{n},\;
    \sin\theta_{n}\sin\varphi_{n}
    \right]^{\rm T},
\end{equation}
where $\theta_{n}$ denotes the zenith angle with respect to the positive $x$-axis, and $\varphi_{n_t}$ denotes the azimuth angle of the projection of $\mathbf p_{n_t}$ onto the $y$-$z$ plane measured from the positive $y$-axis. The overall RA boresight configuration is collected as
\begin{equation}
    \mathbf D=
    \left[
    \mathbf p_1,\mathbf p_2,\ldots,\mathbf p_{N_t}
    \right]
    \in\mathbb R^{3\times N_t}.
\end{equation}
To account for practical rotation constraints \cite{Kumar2023}, the zenith angle of each RA is restricted by
\begin{equation}
    0\leq \theta_{n}\leq \theta_{\max},
    \quad n=1,\ldots,N_t,
\end{equation}
where $\theta_{\max}\in[0,\pi/2]$ is the maximum allowable zenith rotation angle.

\subsection{Channel Model}
We consider a near-field quasi-static channel model, where all objects are located in the near-field region of the BS. Different from the far-field plane-wave model, the spherical wavefront is characterized by the exact propagation distance between each antenna element and each spatial point. For a generic spatial point, its position is parameterized as
\begin{equation}
    \mathbf q(r,\vartheta,\varphi)
    =
    \left[
    r\sin\vartheta\cos\varphi,\;
    r\sin\vartheta\sin\varphi,\;
    r\cos\vartheta
    \right]^{\rm T},
\end{equation}
where $r$ denotes the distance from the TP center, $\vartheta\in[0,\pi]$ denotes the zenith angle with respect to the positive $z$-axis, and $\varphi\in[-\pi,\pi)$ denotes the azimuth angle on the $x$-$y$ plane measured from the positive $x$-axis.
The distance from the $n$-th RA to the spatial point $\mathbf q$ is given by
\begin{equation}
    r_{n}^{\rm t}(\mathbf q)
    =
    \left\|
    \mathbf q-\mathbf t_{n}^{\rm t}
    \right\|_2,
    \quad n=1,\ldots,N_t.
\end{equation}
The corresponding propagation direction is
\begin{equation}
    \mathbf u_{n}(\mathbf q)
    =
    \frac{\mathbf q-\mathbf t_{n}^{\rm t}}
    {r_{n}^{\rm t}(\mathbf q)}.
\end{equation}

The effective antenna gain of each RA depends on both its boresight orientation
and the signal propagation direction, as illustrated in Fig.~\eqref{a}. We adopt the following directional gain
pattern for each RA \cite{Balanis2015}:
\begin{equation} \label{G_origin}
G(\varepsilon  )=
\begin{cases}
G_0 \cos^{2p}\!(\varepsilon ), & \varepsilon \in\left[0,\frac{\pi}{2}\right],\\
0, & \text{otherwise},
\end{cases}
\end{equation}
where $G_0=2(2p+1)$ is the maximum boresight gain to satisfy the law of power conservation, and $p\geq 0$ is the directivity factor that characterizes the beamwidth of the antenna main lobe. Accordingly, the directional gain of the $n$-th RA toward $\mathbf q$ is
given by
\begin{equation}
    G_{n}(\mathbf q,\mathbf D)
    =
    G_0
    \left[
    \mathbf p_{n}^{\rm T}
    \mathbf u_{n}(\mathbf q)
    \right]_+^{2p},
    \label{G_nt_q}
\end{equation}
where $[x]_+\triangleq\max\{x,0\}$, and
$\mathbf p_{n}^{\rm T}\mathbf u_{n}(\mathbf q)
=\cos\psi_{n}(\mathbf q)$. 
For notational convenience, define the RA directional gain vector toward $\mathbf q$ as
\begin{equation}
    \mathbf g(\mathbf q,\mathbf D)
    =
    \left[
    \sqrt{G_1(\mathbf q,\mathbf D)},\ldots,
    \sqrt{G_{N_t}(\mathbf q,\mathbf D)}
    \right].
\end{equation}

\textit{1) Channel model for communication users:}
We adopt a general multipath channel for communication users, comprising one LoS path and $L_k$ NLoS paths originating from environmental scatterers \cite{Lu2023}. Accordingly, the near-field channel between the TP and user $k$ is given by
\begin{equation}
\begin{aligned}
    \mathbf h_k^{\rm H}
    =
    &\beta_{{\rm L},k}
    \mathbf g(\mathbf q_k,\mathbf D)
    \odot
    \mathbf b^{\rm H}(\mathbf q_k) \\
    &+
    \sum_{\ell=1}^{L_k}
    \beta_{{\rm NL},k,\ell}
    \mathbf g(\mathbf q_{k,\ell},\mathbf D)
    \odot
    \mathbf b^{\rm H}(\mathbf q_{k,\ell}),
\end{aligned}
\end{equation}
where $\beta_{{\rm L},k} = \frac{\lambda}{4\pi r_k}e^{-j\frac{2\pi r_k}{\lambda}}$ and $\beta_{{\rm NL},k,\ell}$ are the complex-valued channel gains of the LoS path and the $\ell$-th NLoS path between the TP and user $k$, respectively. Here, $\mathbf q_k=\mathbf q(r_k,\vartheta_k,\varphi_k)$ denotes the position of user $k$, while $\mathbf q_{k,\ell}=\mathbf q(r_{k,\ell},\vartheta_{k,\ell},\varphi_{k,\ell})$ denotes the position of the effective scatterer associated with the $\ell$-th NLoS path.
The near-field channel steering vector $\mathbf b(\mathbf q)$ is given by
\begin{equation}
    \mathbf b^H(\mathbf q)
    =
    \left[
    e^{-j\frac{2\pi}{\lambda}\left(r_1^{\rm t}(\mathbf q)-r\right)},
    \ldots,
    e^{-j\frac{2\pi}{\lambda}\left(r_{N_t}^{\rm t}(\mathbf q)-r\right)}
    \right],
\end{equation}
where $r$ is the distance from the TP center to $\mathbf q$.

\textit{2) Sensing model for target:}
For target sensing, we only consider the LoS path. In the monostatic sensing mode, target sensing relies on the echo signal reflected by the target and received at the RP. The round-trip channel $\mathbf H_s\in\mathbb C^{N_r\times N_t}$ accounts for the propagation path from the TP to the target and from the target back to the RP, which is modeled as
\begin{equation} \label{target echo}
    \mathbf H_s
    =
    \beta_s
    \mathbf b_r(\mathbf q_s)
    \left(
    \mathbf g(\mathbf q_s,\mathbf D)
    \odot
    \mathbf b^{\rm H}(\mathbf q_s)
    \right),
\end{equation}
where $\beta_s=\sqrt{\frac{\lambda^2 \alpha_s}{64\pi^3 r_s^4}}\exp(-j\frac{4\pi r_s}{\lambda})$ is the round-trip complex-valued path coefficient, $\alpha_s$ is the radar cross section (RCS) of target, and $\mathbf q_s=\mathbf q(r_s,\vartheta_s,\varphi_s)$ denotes the target position. Here, $\mathbf b_r(\mathbf q_s)$ is the near-field receive steering vector of the RP toward the target.
The receive steering vector of the RP toward target $\mathbf q_s$ is given by
\begin{equation}
    \mathbf b_r(\mathbf q_s)
    =
    \left[
    e^{-j\frac{2\pi}{\lambda}\left(r_1^{\rm r}(\mathbf q_s)-r_s\right)},
    \ldots,
    e^{-j\frac{2\pi}{\lambda}\left(r_{N_r}^{\rm r}(\mathbf q_s)-r_s\right)}
    \right]^{\rm T},
\end{equation}
where $r_{n_r}^{\rm r}(\mathbf q_s)$ denotes the distance between the $n_r$-th receive antenna and the target $\mathbf q_s$.

\subsection{Signal Model}
Let $\mathbf s \in \mathbb C^{(K+1) \times 1}$ be the independent and identically distributed (i.i.d.) signal vector with $\mathbb E(\mathbf {ss}^\mathrm H) = \mathbf I_{\mathrm K+1}$, where $s_1,\dots,s_K$ are intended for $K$ users, respectively, and $s_{K+1}$ is dedicated for sensing. Consider that the transmitted ISAC waveform is known at the BS, and thus the total $K+1$ data streams can be utilized for sensing. Assume that the BS perfectly knows the channel state information (CSI), where the CSI can be obtained by channel estimation method for RA-aided systems in \cite{Xiong20252}. The digital beamforming matrix is
\begin{equation}
    \mathbf{W} = [ \underbrace{\mathbf{w}_1, \dots, \mathbf{w}_K}_{{\text{for communication and sensing}}}, \underbrace{\mathbf{w}_{K+1}}_{{\text{dedicated for sensing}}}] \in \mathbb{C}^{B \times (K+1)},
\end{equation}
where $w_k \in \mathbb C^{B \times 1}$ denotes the digital beamforming vector of the $k$-th stream.
Under the sub-connected architecture, the analog precoding matrix $\mathbf F$ is expressed as
\begin{equation}
\begin{aligned}
    \mathbf F=\mathrm {Bdiag}(\mathbf f_1,\mathbf f_2,\ldots,\mathbf f_B)\in \mathbb C^{N_t \times B},
\end{aligned}
\end{equation}
where $\mathbf f_b \in \mathbb C^{M \times 1},b=1,2,\dots,B$, is the analog beamforming vector of the $b$-th transmit subarray. Each element of $\mathbf f_b$ satisfies the constant-modulus constraint, i.e, $|\mathbf f_b(j)|=1,j=1,\dots,M$.

The received signal at the $k$-th user is given by
\begin{equation}
\begin{aligned}
    y_k &= \mathbf h^H_k \mathbf {FWs}+n_k \\
    &=\underbrace{\mathbf h^H_k \mathbf F \mathbf w_k s_k}_{\text{desired signal}} + \underbrace{\sum^{K+1}_{j=1,j\neq k}\mathbf h^H_k \mathbf F \mathbf w_j s_j}_{\text{interference}} +n_k,
\end{aligned}
\end{equation}
where $n_k \sim \mathcal{C} \mathcal{N} (0,\sigma_k^2)$ is zero-mean additive white Gaussian noise (AWGN) with noise power $\sigma_k^2$.
Therefore, the SINR of user $k$ is
\begin{equation}
\begin{aligned}
    \text{SINR}_k=\frac{|\mathbf h^H_k \mathbf {Fw}_k|^2}{\sum^{K+1}_{j=1,j\neq k}|\mathbf h^H_k \mathbf {Fw}_j|^2+\sigma_k^2}.
\end{aligned}
\end{equation}
The achievable rate of user $k$ can be obtained as
\begin{equation}\label{Rk}
\begin{aligned}
    R_k=\mathrm {log}_2(1+ \rm{SINR}_k).
\end{aligned}
\end{equation}

For sensing, the RP applies a receive beamformer $\mathbf u \in \mathbb C^{N_r \times 1}$ with $\|\mathbf u\|^2_2=1$ to capture the reflected signal echo. The received signal at the BS can be expressed as
\begin{equation}
    y_s= \underbrace{\mathbf u^\mathrm H \mathbf H_s \mathbf{F W s}}_{\text {target reflection}}+\underbrace{\sum^C_{c=1} \mathbf u^\mathrm H \mathbf H_c \mathbf{F W s}}_{\text {clutter reflections}}+n_s,
\end{equation}
where $n_s \sim \mathcal{C} \mathcal{N} (0,\sigma_s^2)$ denotes the AWGN for radar link. Here, $\mathbf H_s$ and $\mathbf H_c$ denote the round-trip channel matrices associated with the target and the $c$-th clutter, respectively.

We assume the target and clutters are modeled as point-like scatterers, the sensing signal-to-clutter-plus-noise-ratio (SCNR) at the RP is given by
\begin{equation}\label{SCNR}
    \text{SCNR}
    =
    \frac{
    \left\|\mathbf u^\mathrm H \mathbf H_s \mathbf{F W}\right\|_2^2
    }{
    \sum_{c=1}^{C}
    \left\|\mathbf u^\mathrm H \mathbf H_c \mathbf{F W}\right\|_2^2
    +
    \sigma_s^2
    }.
\end{equation}

For point target detection in MIMO radar systems, the SCNR characterizes the target echo strength relative to clutter and noise after receive combining. To further quantify how much information about the target response can be extracted from the received echo under a given transmit waveform, we adopt the sensing mutual information (MI) per unit time as the sensing performance metric~\cite{Ouyang2023,Peng2024}. Specifically, conditioned on the known transmitted signal $\mathbf x$ and the receive beamformer $\mathbf u$, the sensing rate is defined as~\cite{Lyu2025}
\begin{equation}\label{Rs}
    R_s
    =
    I(y_s;\mathbf H_s|\mathbf u,\tilde {\mathbf x})
    =
    \log_2(1+\mathrm{SCNR}),
\end{equation}
where $\tilde {\mathbf x}=\mathbf F\mathbf W\mathbf s$ denotes the transmitted ISAC signal.
\subsection{Problem Formulation}
We aim to maximize a weighted communication-sensing utility by jointly optimizing the RA boresight directions, hybrid transmit beamformer, and receive sensing beamformer. Accordingly, the optimization problem can be formulated as
\begin{subequations}\label{P_1}
\begin{align}
    \max_{\mathbf u,\mathbf F,\mathbf W,\mathbf D}\quad
    & \mathcal G(\mathbf u,\mathbf F,\mathbf W,\mathbf D)
    =
    \varpi_c\sum_{k=1}^{K}\alpha_kR_k+\varpi_s R_s
    \label{G}\\
    \mathrm{s.t.}\quad
    & \|\mathbf u\|_2^2=1, \label{u_constraint}\\
    & \|\mathbf F\mathbf W\|_{\rm F}^2\leq P, \label{P_constraint}\\
    & \mathbf F\in\mathcal A_F, \label{F_constraint}\\
    & \mathbf D\in\mathcal D. \label{D_constraint}
\end{align}
\end{subequations}
where $\varpi_c \geq 0$ and $\varpi_s\geq 0$ are the communication and sensing weights satisfying $\varpi_c+\varpi_s=1$, and $\alpha_k\geq0$ is the user-level rate weight with $\sum_{k=1}^{K}\alpha_k=1$. In this paper, we adopt $\alpha_k=1/K$, and thus $\sum_{k=1}^{K}\alpha_kR_k$ becomes the average communication rate. Constraint \eqref{u_constraint} normalizes the receive beamformer. Constraint \eqref{P_constraint} limits the total transmit power. Constraint \eqref{F_constraint} corresponds to the sub-connected analog precoding structure, where $\mathcal A_F$ denotes the feasible set of block-diagonal analog precoders with constant-modulus elements. Constraint \eqref{D_constraint} represents the feasible set of RA boresight directions satisfying the maximum zenith-angle constraint.

It is worth mentioning that we focus on a single target in this paper for simplicity. Since the proposed formulation can treat the echoes from other targets as radar interference in the denominator of \eqref{SCNR}, it can be similarly adapted to address the multi-target problem.

\section{Proposed Algorithm}
Problem \eqref{P_1} is difficult to solve due to its non-convex objective, coupled variables, and practical hardware constraints. Specifically, the fractional SINR/SCNR terms couple $\mathbf W$, $\mathbf F$, $\mathbf u$, and $\mathbf D$, while the constant-modulus and RA zenith-angle constraints further complicate the optimization.

In this section, we propose an AO algorithm to tackle these challenges. Specifically, problem \eqref{P_1} is decomposed into several tractable subproblems, where the digital beamformer, analog precoder, receive beamformer, and RA boresight directions are updated alternately while keeping the other variables fixed. Unlike conventional two-stage designs based on fully digital beamformer approximation, the proposed algorithm directly optimizes the hybrid beamforming variables under the practical sub-connected architecture.
For the variables $\mathbf F$, $\mathbf W$, $\mathbf D$, the FP approach is employed\cite{ShenK2018}. We introduce auxiliary variables $\boldsymbol {\mu} =[ \mu_1,\dots,\mu_{K+1}]$, $\boldsymbol {\xi}^c =[ \xi^c_1,\dots, \xi^c_K]^T$ and $\boldsymbol {\xi}^s =[ \xi^s_1,\dots, \xi^s_{K+1}]^T$  to transform \eqref{G} into an equivalent tractable form \eqref{FP}. Based on the AO framework, the problem \eqref{P_1} can be decomposed as five sub-problems. The details of the proposed algorithms are presented below.

\begin{figure*}[b] 
    \centering
    \hrulefill
    \begin{align}
        \tilde{\mathcal{G}} 
        &= \varpi_c\frac{1}{K} \sum_{k=1}^K \text{log}(1 + \mu_k) + \varpi_s \text{log}(1 + \mu_{K+1}) - \varpi_c\frac{1}{K} \sum_{k=1}^K \mu_k - \varpi_s \mu_{K+1} \notag \\
        & + \varpi_c\frac{1}{K} \sum_{k=1}^K \Bigl( 2\sqrt{1 + \mu_k} \mathrm{Re}\bigl\{ \xi_k^c \mathbf{h}_k^H(\mathbf D) \mathbf F\mathbf{w}_k \bigr\} 
        - \lvert \xi_k^c \rvert^2 \Bigl( \sum_{j=1}^{K+1} \lvert \mathbf{h}_k^H(\mathbf D) \mathbf F\mathbf{w}_j \rvert^2 + \sigma_k^2 \Bigr) \Bigr) \notag \\
        & + \varpi_s \Bigl( 2\sqrt{1 + \mu_{K+1}} \mathrm{Re}\bigl\{ \mathbf u^H \mathbf{H}_s(\mathbf D) \mathbf{F} \mathbf W \boldsymbol{\xi}^s \bigr\} 
        - \lVert \boldsymbol{\xi}^s \rVert^2 \Bigl( \sum_{c=1}^C \lVert  \mathbf u^H \mathbf{H}_c(\mathbf D) \mathbf{F} \mathbf{W}\rVert^2 
        + \lVert \mathbf u^H \mathbf{H}_s(\mathbf D) \mathbf{F} \mathbf{W}\rVert^2 + \sigma_s^2 \Bigr) \Bigr)  
        \label{FP}
    \end{align}
\end{figure*}
\subsection{Receive Beamforming Optimization}
Given $\{\mathbf F,\mathbf W,\mathbf D\}$, the optimization of $\mathbf u$
only affects the SCNR. Therefore, maximizing the objective value in \eqref{G} with respect to
$\mathbf u$ is equivalent to maximizing the SCNR. Hence, the subproblem for
$\mathbf u$ is formulated as
\begin{subequations}\label{u_optimize}
    \begin{align}
    \max_{\mathbf u}\quad & \mathrm{SCNR} \label{P_u}\\
    \mathrm{s.t.}\quad & \eqref{u_constraint}.
    \end{align}
\end{subequations}

Based on the SCNR expression \eqref{SCNR}, problem \eqref{u_optimize} can be regarded as
a generalized Rayleigh quotient maximization problem. Since the target is
modeled as a point-like scatterer, the target response covariance is rank-one.
By applying Proposition~\ref{Proposition:1}, the optimal closed-form solution
can be obtained.

\begin{proposition}\label{Proposition:1}
For given $\{\mathbf F,\mathbf W,\mathbf D\}$, the optimal solution to
problem \eqref{u_optimize} is given by
\begin{equation}
    \mathbf u^\star
    =
    \frac{
    \left(
    \sum_{c=1}^{C}
    \tilde{\mathbf H}_c\tilde{\mathbf H}_c^{\rm H}
    +
    \sigma_s^2\mathbf I_{N_r}
    \right)^{-1}
    \mathbf b_r(\mathbf q_s)
    }{
    \left\|
    \left(
    \sum_{c=1}^{C}
    \tilde{\mathbf H}_c\tilde{\mathbf H}_c^{\rm H}
    +
    \sigma_s^2\mathbf I_{N_r}
    \right)^{-1}
    \mathbf b_r(\mathbf q_s)
    \right\|_2
    },
    \label{eq:u_optimal}
\end{equation}
where $\tilde{\mathbf H}_c = \mathbf H_c\mathbf F\mathbf W
\in \mathbb C^{N_r\times (K+1)}$ denotes the $c$-th clutter response, and $\mathbf b_r(\mathbf q_s)$ is the receive steering vector
toward the target.
\end{proposition}

\begin{proof}
See Appendix A.
\end{proof}
    
\subsection{Transmit Digital Beamforming Optimization}
Due to the sub-connected structure and the constant-modulus
constraint of $\mathbf F$, we have $\mathbf F^{\rm H}\mathbf F=M\mathbf I_B$. Therefore, the constraint \eqref{P_constraint} can be rewritten as
\begin{equation}
    \|\mathbf F\mathbf W\|_{\rm F}^2
    =
    \mathrm{Tr}(\mathbf W^{\rm H}\mathbf F^{\rm H}\mathbf F\mathbf W)
    =
    M\|\mathbf W\|_{\rm F}^2 .
\end{equation}
Thus, the digital beamforming subproblem is formulated as
\begin{subequations}\label{P_W}
\begin{align}
    \max_{\mathbf W}\quad
    & \tilde{\mathcal G}(\mathbf W)\\
    \mathrm{s.t.}\quad
    & \|\mathbf W\|_{\rm F}^2\leq \frac{P}{M}.
\end{align}
\end{subequations}

For fixed $\{\mathbf u,\mathbf F,\mathbf D,\boldsymbol\mu,
\boldsymbol\xi^s,\boldsymbol\xi^c\}$, the objective function
$\tilde{\mathcal G}(\mathbf W)$ is a concave quadratic function with respect
to $\mathbf W$. Specifically, the terms related to $\mathbf W$ can be written as
\begin{equation}
    \tilde{\mathcal G}(\mathbf W)
    =\sum_{k=1}^{K+1}
    \left(
    2\mathrm{Re}\{\boldsymbol{\varphi}_k^{\rm H}\mathbf w_k\}
    -
    \mathbf w_k^{\rm H}\boldsymbol{\Lambda}\mathbf w_k
    \right)+{\rm const},
\end{equation}
where
\begin{align}
\boldsymbol{\Lambda}
&=
\varpi_c\frac{1}{K}
\sum_{i=1}^{K}
|\xi_i^c|^2
(\mathbf h_i^{\rm H}\mathbf F)^{\rm H}
(\mathbf h_i^{\rm H}\mathbf F)
\notag\\
&\quad+
\varpi_s
\|\boldsymbol{\xi}^s\|_2^2
\bigg[
(\mathbf u^{\rm H}\mathbf H_s\mathbf F)^{\rm H}
(\mathbf u^{\rm H}\mathbf H_s\mathbf F)
\notag\\
&\qquad+
\sum_{c=1}^{C}
(\mathbf u^{\rm H}\mathbf H_c\mathbf F)^{\rm H}
(\mathbf u^{\rm H}\mathbf H_c\mathbf F)
\bigg],
\end{align}
and
\begin{align}
\boldsymbol{\varphi}_k
&=
\varpi_c\frac{1}{K}\sqrt{1+\mu_k}
(\xi_k^c\mathbf h_k^{\rm H}\mathbf F)^{\rm H}
\notag\\
&+
\varpi_s\sqrt{1+\mu_{K+1}}
(\xi_k^s\mathbf u^{\rm H}\mathbf H_s\mathbf F)^{\rm H},
 k=1,\ldots,K,
\\
\boldsymbol{\varphi}_{K+1}
&=
\varpi_s\sqrt{1+\mu_{K+1}}
(\xi_{K+1}^s\mathbf u^{\rm H}\mathbf H_s\mathbf F)^{\rm H}.
\end{align}

Since $\tilde{\mathcal G}(\mathbf W)$ is concave in $\mathbf W$ and the feasible set is convex, problem \eqref{P_W} can be globally solved via the KKT conditions. To handle the transmit power constraint, we introduce a
nonnegative Lagrange multiplier $\lambda$ and construct the following
Lagrangian:
\begin{equation}
    \mathcal L(\mathbf W,\lambda)
    =
    -\tilde{\mathcal G}(\mathbf W)
    +
    \lambda
    \left(
    \|\mathbf W\|_{\rm F}^2-\frac{P}{M}
    \right),
\end{equation}
where $\lambda\geq0$ is the Lagrange multiplier corresponding to the power constraint.
By setting the Wirtinger derivative of $\mathcal L(\mathbf W,\lambda)$ with
respect to $\mathbf w_k^\ast$ to zero, the closed-form expression of $\mathbf W$ can be obtained, which is
\begin{equation}
    \mathbf w_k(\lambda)
    =
    \left(
    \boldsymbol{\Lambda}
    +
    \lambda\mathbf I_B
    \right)^{-1}
    \boldsymbol{\varphi}_k,
    \quad k=1,\ldots,K+1.
    \label{eq:w_lambda}
\end{equation}
If $\|\mathbf W(0)\|_{\rm F}^2\leq P/M$, then $\lambda^\star=0$. Otherwise, $\lambda$ needs to be decided to satisfy
\begin{equation}
    h(\lambda)
    =
    \|\mathbf W(\lambda)\|_{\rm F}^2
    -
    \frac{P}{M}
    \leq \varepsilon .
\end{equation}
Since $\boldsymbol{\Lambda}\succeq \mathbf 0$, increasing $\lambda$ enhances
the regularization of
$(\boldsymbol{\Lambda}+\lambda\mathbf I_B)^{-1}$ and thus monotonically
decreases $\|\mathbf W(\lambda)\|_{\rm F}^2$. Therefore, $h(\lambda)$ is
monotonically decreasing with respect to $\lambda$, and $\lambda^\star$ can be
found via the bisection method \cite{PanC2020}.

\subsection{Transmit Analog Beamforming Optimization}
Given $\{\mathbf u,\mathbf W,\mathbf D,\boldsymbol\mu,
\boldsymbol\xi^s,\boldsymbol\xi^c\}$, we optimize the analog beamformer
$\mathbf F$. The corresponding subproblem is
\begin{subequations}\label{P_F}
\begin{align}
    \max_{\mathbf F}\quad
    & \tilde{\mathcal G}(\mathbf F)\\
    \mathrm{s.t.}\quad
    & \mathbf F\in\mathcal A_F .
\end{align}
\end{subequations}

The difficulty of \eqref{P_F} lies in the block-diagonal unit-modulus
structure of $\mathbf F$. Since only its nonzero entries are optimizable, we
collect them into a compact vector
\begin{equation}
    \mathbf z
    =
    [\mathbf f_1^{\rm T},\ldots,\mathbf f_B^{\rm T}]^{\rm T}
    \in\mathbb C^{N_t\times 1}.
\end{equation}
Then, for any digital beamformer $\mathbf w_j$, we have
\begin{equation}
    \mathbf F\mathbf w_j
    =
    \left(
    \mathrm{diag}(\mathbf w_j)\otimes\mathbf I_M
    \right)
    \mathbf z .
    \label{Fw_z}
\end{equation}
For the sensing-related term $\mathbf F\mathbf W$, we further exploit the SVD
of $\mathbf W$ as
\begin{equation}
    \mathbf W
    =
    \sum_{r=1}^{R}
    \rho_r\mathbf u_r\mathbf v_r^{\rm H},
\end{equation}
where $R=\mathrm{rank}(\mathbf W)$, $\rho_r$ is the $r$-th singular value, and
$\mathbf u_r$ and $\mathbf v_r$ are the corresponding left and right singular
vectors. By defining
$\tilde{\mathbf u}_r=\sqrt{\rho_r}\mathbf u_r$ and
$\tilde{\mathbf v}_r=\sqrt{\rho_r}\mathbf v_r$, we obtain
\begin{equation}
    \mathbf F\mathbf W
    =
    \sum_{r=1}^{R}
    \left(
    \mathrm{diag}(\tilde{\mathbf u}_r)\otimes\mathbf I_M
    \right)
    \mathbf z
    \tilde{\mathbf v}_r^{\rm H}.
    \label{FW_svd_z}
\end{equation}
This SVD-based compact vector reformulation preserves the sub-connected
hardware structure, removes the structurally zero entries, and makes the
objective explicitly differentiable with respect to $\mathbf z$.

Based on \eqref{Fw_z} and \eqref{FW_svd_z}, problem \eqref{P_F} can be
equivalently recast as
\begin{subequations}\label{F_optimize}
\begin{align}
    \max_{\mathbf z}\quad
    &
    \tilde{\mathcal P}(\mathbf z)
    =
    2\mathrm{Re}\{\tilde{\boldsymbol\beta}^{\rm H}\mathbf z\}
    -
    \mathbf z^{\rm H}\tilde{\boldsymbol\Xi}\mathbf z
    +
    {\rm const}
    \label{eq:F_optimize}\\
    \mathrm{s.t.}\quad
    &
    |[\mathbf z]_n|=1,\quad n=1,\ldots,N_t .
    \label{z_constraint}
\end{align}
\end{subequations}
Here, $\tilde{\boldsymbol\beta}$ and $\tilde{\boldsymbol\Xi}$ collect the
linear and quadratic coefficients with respect to $\mathbf z$, respectively.
Specifically,
\begin{equation}
    \tilde{\boldsymbol\beta}
    =
    \varpi_c\frac{1}{K}\sum_{k=1}^{K}\tilde{\boldsymbol\beta}_k
    +
    \varpi_s\tilde{\boldsymbol\beta}_s,
\end{equation}
where
\begin{equation}
    \tilde{\boldsymbol\beta}_k^{\rm H}
    =
    \sqrt{1+\mu_k}\xi_k^c
    \mathbf h_k^{\rm H}
    \left(
    \mathrm{diag}(\mathbf w_k)\otimes\mathbf I_M
    \right),
\end{equation}
and
\begin{equation}
    \tilde{\boldsymbol\beta}_s^{\rm H}
    =
    \sqrt{1+\mu_{K+1}}
    \sum_{r=1}^{R}
    \left(
    \tilde{\mathbf v}_r^{\rm H}\boldsymbol\xi^s
    \right)
    \mathbf u^{\rm H}\mathbf H_s
    \left(
    \mathrm{diag}(\tilde{\mathbf u}_r)\otimes\mathbf I_M
    \right).
\end{equation}
The quadratic coefficient matrix is given by
\begin{align}
    \tilde{\boldsymbol\Xi}
    &=
    \varpi_c\frac{1}{K}
    \sum_{k=1}^{K}
    |\xi_k^c|^2
    \sum_{j=1}^{K+1}
    \tilde{\mathbf h}_{k,j}
    \tilde{\mathbf h}_{k,j}^{\rm H}
    \notag\\
    &\quad+
    \varpi_s
    \|\boldsymbol\xi^s\|_2^2
    \left(
    \sum_{c=1}^{C}
    \sum_{r=1}^{R}
    \tilde V_r
    \tilde{\mathbf g}_{c,r}
    \tilde{\mathbf g}_{c,r}^{\rm H}
    +
    \sum_{r=1}^{R}
    \tilde V_r
    \tilde{\mathbf g}_{s,r}
    \tilde{\mathbf g}_{s,r}^{\rm H}
    \right),
\end{align}
with
\begin{equation}
    \tilde{\mathbf h}_{k,j}^{\rm H}
    =
    \mathbf h_k^{\rm H}
    \left(
    \mathrm{diag}(\mathbf w_j)\otimes\mathbf I_M
    \right),
\end{equation}
\begin{equation}
    \tilde{\mathbf g}_{c,r}^{\rm H}
    =
    \mathbf u^{\rm H}\mathbf H_c
    \left(
    \mathrm{diag}(\tilde{\mathbf u}_r)\otimes\mathbf I_M
    \right),
\end{equation}
\begin{equation}
    \tilde{\mathbf g}_{s,r}^{\rm H}
    =
    \mathbf u^{\rm H}\mathbf H_s
    \left(
    \mathrm{diag}(\tilde{\mathbf u}_r)\otimes\mathbf I_M
    \right),
\end{equation}
and $\tilde V_r=\|\tilde{\mathbf v}_r\|_2^2=\rho_r$.

The reformulated problem \eqref{F_optimize} has a compact unit-modulus vector
structure and can be efficiently solved via manifold optimization. Accordingly,
its feasible set is characterized by the complex circle manifold
\begin{equation}
    \mathcal M
    =
    \left\{
    \mathbf z\in\mathbb C^{N_t\times 1}
    \mid
    |[\mathbf z]_n|=1,\; n=1,\ldots,N_t
    \right\}.
\end{equation}
The tangent space of $\mathcal M$ at $\mathbf z$ is given by
\begin{equation}
    T_{\mathbf z}\mathcal M
    =
    \left\{
    \boldsymbol\zeta\in\mathbb C^{N_t\times 1}
    \mid
    \mathrm{Re}\{
    \boldsymbol\zeta\odot \mathbf z^*
    \}
    =
    \mathbf 0
    \right\},
\end{equation}
where $\boldsymbol\zeta$ denotes a tangent vector at $\mathbf z$.

The Euclidean gradient of the objective function $\tilde{\mathcal P}(\mathbf z)$
with respect to the complex vector $\mathbf z$ is given by
\begin{equation}
    \nabla_{\mathbf z}\tilde{\mathcal P}(\mathbf z)
    =
    \tilde{\boldsymbol\beta}
    -
    \tilde{\boldsymbol\Xi}\mathbf z .
    \label{grad_z}
\end{equation}
Then, the Riemannian gradient is obtained by projecting the Euclidean gradient
onto the tangent space, i.e.,
\begin{equation}
    \mathrm{grad}\,\tilde{\mathcal P}(\mathbf z)
    =
    \nabla_{\mathbf z}\tilde{\mathcal P}(\mathbf z)
    -
    \mathrm{Re}
    \left\{
    \nabla_{\mathbf z}\tilde{\mathcal P}(\mathbf z)
    \odot
    \mathbf z^*
    \right\}
    \odot
    \mathbf z .
    \label{Rgrad_z}
\end{equation}

Then, problem (42) is solved by a Riemannian conjugate gradient
method with Armijo backtracking \cite{Yu2016}, where the unit-modulus constraint
is maintained by the standard normalization retraction
\[
\mathrm{Retr}_{\mathbf z}(\alpha\boldsymbol\zeta)
=
\frac{\mathbf z+\alpha\boldsymbol\zeta}
{|\mathbf z+\alpha\boldsymbol\zeta|}.
\]
After convergence, the analog beamformer \(\mathbf F\) is reconstructed
by placing the optimized entries of \(\mathbf z\) back into the
corresponding block-diagonal positions.

\subsection{Auxiliary Variables Optimization}

For given $\{\mathbf u,\mathbf F,\mathbf W,\mathbf D\}$, the auxiliary
variables $\boldsymbol\mu$, $\boldsymbol\xi^c$, and $\boldsymbol\xi^s$ can be
updated in closed form. By setting the corresponding first-order optimality
conditions of $\tilde{\mathcal G}$ to zero, we obtain
\begin{equation}\label{xic}
    \xi_k^c
    =
    \frac{
    \sqrt{1+\mu_k}
    \left(\mathbf h_k^{\rm H}\mathbf F\mathbf w_k\right)^*
    }{
    \sum_{j=1}^{K+1}
    \left|
    \mathbf h_k^{\rm H}\mathbf F\mathbf w_j
    \right|^2
    +
    \sigma_k^2
    },
    \quad k=1,\ldots,K,
\end{equation}
and
\begin{equation}\label{xis}
    \boldsymbol\xi^s
    =
    \frac{
    \sqrt{1+\mu_{K+1}}
    \left(
    \mathbf u^{\rm H}\mathbf H_s\mathbf F\mathbf W
    \right)^{\rm H}
    }{
    \sum_{c=1}^{C}
    \left\|
    \mathbf u^{\rm H}\mathbf H_c\mathbf F\mathbf W
    \right\|_2^2
    +
    \left\|
    \mathbf u^{\rm H}\mathbf H_s\mathbf F\mathbf W
    \right\|_2^2
    +
    \sigma_s^2
    }.
\end{equation}

The auxiliary variables introduced by the Lagrangian dual transform are updated as
\begin{equation}\label{muk}
    \mu_k
    =
    \frac{
    \left|
    \mathbf h_k^{\rm H}\mathbf F\mathbf w_k
    \right|^2
    }{
    \sum_{j=1,j\neq k}^{K+1}
    \left|
    \mathbf h_k^{\rm H}\mathbf F\mathbf w_j
    \right|^2
    +
    \sigma_k^2
    },
    \quad k=1,\ldots,K,
\end{equation}
and
\begin{equation}\label{muK+1}
    \mu_{K+1}
    =
    \frac{
    \left\|
    \mathbf u^{\rm H}\mathbf H_s\mathbf F\mathbf W
    \right\|_2^2
    }{
    \sum_{c=1}^{C}
    \left\|
    \mathbf u^{\rm H}\mathbf H_c\mathbf F\mathbf W
    \right\|_2^2
    +
    \sigma_s^2
    }.
\end{equation}

\subsection{RA Boresight Optimization}
Given $\{\mathbf u,\mathbf F,\mathbf W,\boldsymbol\mu,
\boldsymbol\xi^c,\boldsymbol\xi^s\}$, we optimize the element-wise RA
boresight configuration $\mathbf D=[\mathbf p_1,\ldots,\mathbf p_{N_t}]$. The RA boresight optimization subproblem is formulated as
\begin{subequations}\label{P_D}
\begin{align}
    \max_{\mathbf D}\quad
    & \tilde{\mathcal G}(\mathbf D)\\
    \mathrm{s.t.}\quad
    & \|\mathbf p_{n}\|_2=1,\quad n=1,\ldots,N_t,\label{P_D_1}\\
    & \cos\theta_{\max} \leq \mathbf p_{n}^{\rm T}\mathbf e_0\leq 1, 
    \quad n=1,\ldots,N_t, \label{P_D_2}
\end{align}
\end{subequations}
where $\mathbf e_0=[1,0,0]^{\rm T}$ denotes the outward normal direction of the TP.
The constraints in \eqref{P_D} characterize a spherical cap for each RA boresight, which accounts for the practical rotation range of each RA.

Problem \eqref{P_D} is non-convex for two main reasons: (i) the RA boresight
directions enter the communication and sensing channels through the
directional-gain factors, leading to nonlinear coupling with the hybrid
beamforming variables; and (ii) constraints \eqref{P_D_1}--\eqref{P_D_2}
restrict each RA boresight to a compact but nonconvex spherical-cap
feasible region. To address this challenge, we develop a 
Frank--Wolfe--based algorithm tailored to the spherical-cap constraint \cite{Peng2026}. 
The proposed update directly optimizes the boresight direction of each RA using its element gradient and a closed-form cap-constrained linear oracle.

Define the feasible spherical cap for each RA boresight as
\begin{equation}
    \mathcal C_{\rm cap}
    \triangleq
    \left\{
    \mathbf x\in\mathbb R^3:
    \|\mathbf x\|_2=1,\;
    \mathbf x^{\rm T}\mathbf e_0\geq \cos\theta_{\max}
    \right\}.
\end{equation}
Accordingly, the feasible set of problem \eqref{P_D} is the Cartesian product of $N_t$ spherical caps:
\begin{equation}
    \mathcal C
    \triangleq
    \prod_{n=1}^{N_t}\mathcal C_{\rm cap}.
\end{equation}

To implement the proposed Frank--Wolfe--based update, we first need the
gradient of $\tilde{\mathcal G}$ with respect to each RA boresight vector.
For the $n_t$-th RA, this gradient is defined component-wise as
\begin{equation}
    \frac{\partial \tilde{\mathcal G}}
    {\partial \mathbf p_{n_t}}
    =
    \left[
    \frac{\partial \tilde{\mathcal G}}
    {\partial [\mathbf p_{n_t}]_1},
    \frac{\partial \tilde{\mathcal G}}
    {\partial [\mathbf p_{n_t}]_2},
    \frac{\partial \tilde{\mathcal G}}
    {\partial [\mathbf p_{n_t}]_3}
    \right]^{\rm T}.
\end{equation}
The partial derivatives $\partial \tilde{\mathcal G}/\partial \mathbf p_{n_t}$ are derived in
\eqref{eq:dh_entry_dp}--\eqref{eq:dG_dp_component}. Therein, $\mathbf e_{n_t}\in\mathbb R^{N_t\times 1}$ denotes the $n_t$-th standard
basis vector, whose $n_t$-th entry is one and all other entries are zero.
\begin{figure*}[b]
    \centering
    \hrulefill
    \begin{align}
    \frac{\partial [\mathbf h_k]_{n_t}}
    {\partial [\mathbf p_{n_t}]_i}
    &=
    p\sqrt{G_0}
    \Bigg(
    \beta_{{\rm L},k}^{*}
    \left[
    \mathbf p_{n_t}^{\rm T}
    \mathbf u_{n_t}(\mathbf q_k)
    \right]_+^{p-1}
    [\mathbf b(\mathbf q_k)]_{n_t}
    [\mathbf u_{n_t}(\mathbf q_k)]_i
    \notag\\
    &\quad+
    \sum_{\ell=1}^{L_k}
    \beta_{{\rm NL},k,\ell}^{*}
    \left[
    \mathbf p_{n_t}^{\rm T}
    \mathbf u_{n_t}(\mathbf q_{k,\ell})
    \right]_+^{p-1}
    e^{j\frac{2\pi}{\lambda}
    \left(r_{n_t}^{\rm t}(\mathbf q_{k,\ell})-r_{k,\ell}\right)}
    [\mathbf u_{n_t}(\mathbf q_{k,\ell})]_i
    \Bigg),
    \quad i=1,2,3,
    \label{eq:dh_entry_dp}
    \\
    \frac{\partial \mathbf H_\chi}
    {\partial [\mathbf p_{n_t}]_i}
    &=
    p\sqrt{G_0}\beta_\chi
    \left[
    \mathbf p_{n_t}^{\rm T}
    \mathbf u_{n_t}(\mathbf q_\chi)
    \right]_+^{p-1}
    [\mathbf u_{n_t}(\mathbf q_\chi)]_i
    e^{-j\frac{2\pi}{\lambda}
    \left(r_{n_t}^{\rm t}(\mathbf q_\chi)-r_\chi\right)}
    \mathbf b_r(\mathbf q_\chi)\mathbf e_{n_t}^{\rm T},
    \quad \chi\in\{s,c\},
    \label{eq:dH_dp}
    \\
    \frac{\partial \tilde{\mathcal G}}
    {\partial [\mathbf p_{n_t}]_i}
    &=
    2\varpi_c\frac{1}{K}
    \sum_{k=1}^{K}
    \mathrm{Re}
    \left\{
    \sqrt{1+\mu_k}\xi_k^c
    \left(
    \frac{\partial \mathbf h_k}
    {\partial [\mathbf p_{n_t}]_i}
    \right)^{\rm H}
    \mathbf F\mathbf w_k
    \right\}
    -
    2\varpi_c\frac{1}{K}
    \sum_{k=1}^{K}
    |\xi_k^c|^2
    \sum_{j=1}^{K+1}
    \mathrm{Re}
    \left\{
    \left(
    \frac{\partial \mathbf h_k}
    {\partial [\mathbf p_{n_t}]_i}
    \right)^{\rm H}
    \mathbf F\mathbf w_j
    \mathbf w_j^{\rm H}\mathbf F^{\rm H}
    \mathbf h_k
    \right\}
    \notag\\
    &\quad+
    2\varpi_s
    \sqrt{1+\mu_{K+1}}
    \mathrm{Re}
    \left\{
    \mathbf u^{\rm H}
    \left(
    \frac{\partial \mathbf H_s}
    {\partial [\mathbf p_{n_t}]_i}
    \right)
    \mathbf F\mathbf W\boldsymbol\xi^s
    \right\}
    -
    2\varpi_s
    \|\boldsymbol\xi^s\|_2^2
    \sum_{c=1}^{C}
    \mathrm{Re}
    \left\{
    \mathbf u^{\rm H}
    \left(
    \frac{\partial \mathbf H_c}
    {\partial [\mathbf p_{n_t}]_i}
    \right)
    \mathbf F\mathbf W\mathbf W^{\rm H}
    \mathbf F^{\rm H}\mathbf H_c^{\rm H}\mathbf u
    \right\}
    \notag\\
    &\quad-
    2\varpi_s
    \|\boldsymbol\xi^s\|_2^2
    \mathrm{Re}
    \left\{
    \mathbf u^{\rm H}
    \left(
    \frac{\partial \mathbf H_s}
    {\partial [\mathbf p_{n_t}]_i}
    \right)
    \mathbf F\mathbf W\mathbf W^{\rm H}
    \mathbf F^{\rm H}\mathbf H_s^{\rm H}\mathbf u
    \right\},
    \quad i=1,2,3.
    \label{eq:dG_dp_component}
    \end{align}
\end{figure*}
At iteration $t$, the Euclidean gradient associated with $n$-th RA boresight can be given by
\begin{equation}
    \mathbf g_{n_t}^{(t)}
=
\frac{\partial \tilde{\mathcal G}}{\partial \mathbf p_{n_t}},
    \quad n_t=1,\ldots,N_t.
    \label{grad_block_D}
\end{equation}

Next, we construct a feasible ascent direction under the spherical cap constraint. Due to the unit-norm constraint $\|\mathbf p_{n_t}\|_2=1$, an arbitrary Euclidean perturbation is not necessarily feasible. For a small perturbation $\Delta\mathbf p_{n_t}$ around the current point $\mathbf p_{n_t}^{(t)}$, the unit-norm constraint requires
$\left\|\mathbf p_{n_t}^{(t)}+\Delta\mathbf p_{n_t}\right\|_2^2=1 $.
By neglecting the second-order term $\|\Delta\mathbf p_{n_t}\|_2^2$, we obtain
$\left(\mathbf p_{n_t}^{(t)}\right)^{\rm T}\Delta\mathbf p_{n_t}=0$.
This indicates that any first-order feasible perturbation should lie in the tangent space of the unit sphere at $\mathbf p_{n_t}^{(t)}$. Therefore, we project the Euclidean gradient $\mathbf g_{n_t}^{(t)}$ onto this tangent space as
\begin{equation}
    \bar{\mathbf g}_{n_t}^{(t)}
    =
    \left(
    \mathbf I_3
    -
    \mathbf p_{n_t}^{(t)}(\mathbf p_{n_t}^{(t)})^{\rm T}
    \right)
    \mathbf g_{n_t}^{(t)}.
    \label{proj_grad_D}
\end{equation}
The projected gradient $\bar{\mathbf g}_{n_t}^{(t)}$ characterizes the locally feasible ascent direction on the unit sphere.

Based on $\bar{\mathbf g}_{n_t}^{(t)}$, we linearize the objective around the current boresight $\mathbf p_{n_t}^{(t)}$ and search for the point on the spherical cap that is most aligned with the projected gradient. This gives the following cap-constrained linear oracle:
\begin{equation}
    \mathbf s_{n_t}^{(t)}
    =
    \arg\max_{\mathbf s\in\mathcal C_{\rm cap}}
    \left(\bar{\mathbf g}_{n_t}^{(t)}\right)^{\rm T}\mathbf s .
    \label{oracle_D}
\end{equation}
The oracle in \eqref{oracle_D} admits a simple geometric interpretation. If the normalized projected gradient lies inside the spherical cap, it is directly selected as the oracle solution. Otherwise, the zenith-angle constraint becomes active, and the optimal point lies on the boundary circle of the cap, whose tangential component is aligned with the projection of $\bar{\mathbf g}_{n_t}^{(t)}$ onto the plane orthogonal to $\mathbf e_0$.

For $\|\bar{\mathbf g}_{n_t}^{(t)}\|_2>0$, define
$\hat{\mathbf g}_{n_t}^{(t)}
    =
    \frac{\bar{\mathbf g}_{n_t}^{(t)}}{\|\bar{\mathbf g}_{n_t}^{(t)}\|_2}$
and
$\hat{\mathbf g}_{{n_t},\perp}^{(t)}
    =
    \hat{\mathbf g}_{n_t}^{(t)}
    -
    \left(
    (\hat{\mathbf g}_{n_t}^{(t)})^{\rm T}\mathbf e_0
    \right)\mathbf e_0$.
Then, the closed-form solution of \eqref{oracle_D} is given in \eqref{oracle_solution_D}, where $\boldsymbol\eta_{n_t}$ is any unit vector satisfying $\boldsymbol\eta_{n_t}^{\rm T}\mathbf e_0=0$. If $\|\bar{\mathbf g}_{n_t}^{(t)}\|_2=0$, no local ascent direction is available, and we set $\mathbf s_{n_t}^{(t)}=\mathbf p_{n_t}^{(t)}$.
\begin{figure*}[b] 
    \centering 
    \hrulefill 
    \begin{equation} 
        \mathbf s_{n_t}^{(t)} = 
        \begin{cases}
        \hat{\mathbf g}_{n_t}^{(t)}, & (\hat{\mathbf g}_{n_t}^{(t)})^{\rm T}\mathbf e_0 \geq \cos(\theta_{\max}), \\
        \cos(\theta_{\max})\mathbf e_0 + \sin(\theta_{\max}) \dfrac{ \hat{\mathbf g}_{{n_t},\perp}^{(t)} }{ \|\hat{\mathbf g}_{{n_t},\perp}^{(t)}\|_2 }, & (\hat{\mathbf g}_{n_t}^{(t)})^{\rm T}\mathbf e_0 < \cos(\theta_{\max}), \; \|\hat{\mathbf g}_{{n_t},\perp}^{(t)}\|_2>0, \\
        \cos(\theta_{\max})\mathbf e_0 + \sin(\theta_{\max})\boldsymbol\eta_{n_t}, & (\hat{\mathbf g}_{n_t}^{(t)})^{\rm T}\mathbf e_0 < \cos(\theta_{\max}), \; \|\hat{\mathbf g}_{{n_t},\perp}^{(t)}\|_2=0, 
        \end{cases} \label{oracle_solution_D} 
    \end{equation} 
    \end{figure*}
The Frank--Wolfe-based direction for the ${n_t}$-th RA is given by
\begin{equation}
    \boldsymbol\Delta_{n_t}^{(t)}
    =
    \mathbf s_{n_t}^{(t)}
    -
    \mathbf p_{n_t}^{(t)}.
    \label{FW_direction_D}
\end{equation}
The corresponding optimality gap is given by
\begin{equation}
    \sigma_{\rm FW}^{(t)}
    =
    \sum_{{n_t}=1}^{{N_t}}
    \left(
    \bar{\mathbf g}_{n_t}^{(t)}
    \right)^{\rm T}
    \boldsymbol\Delta_{n_t}^{(t)}.
    \label{FW_gap_D}
\end{equation}
If $\sigma_{\rm FW}^{(t)}$ is below a prescribed threshold, the boresight optimization is terminated.
Otherwise, the boresight vectors are updated according to
\begin{equation}
    \mathbf p_{n_t}^{(t+1)}
    =
    \frac{
    \mathbf p_{n_t}^{(t)}
    +
    \rho^{(t)}
    \boldsymbol\Delta_{n_t}^{(t)}
    }{
    \left\|
    \mathbf p_{n_t}^{(t)}
    +
    \rho^{(t)}
    \boldsymbol\Delta_{n_t}^{(t)}
    \right\|_2
    },
    \quad {n_t}=1,\ldots,{N_t},
    \label{D_update}
\end{equation}
where step size $\rho^{(t)}\in(0,1]$ is determined by Armijo backtracking. Specifically, $\rho^{(t)}$ is accepted if
\begin{equation}
    \tilde{\mathcal G}(\mathbf D^{(t+1)})
    \geq
    \tilde{\mathcal G}(\mathbf D^{(t)})
    +
    c_{\rm A}\rho^{(t)}\sigma_{\rm FW}^{(t)},
    \label{Armijo_D}
\end{equation}
where $c_{\rm A}\in(0,1)$ is the Armijo parameter. This sufficient-increase condition guarantees that each accepted update does not decrease the transformed objective.

\begin{algorithm}[!t]
    \caption{Proposed Frank--Wolfe-Based Algorithm for RA Boresight Optimization}
    \label{alg:D_update}
\begin{algorithmic}[1]
    \State \textbf{Initialize:} feasible boresight configuration $\mathbf D^{(0)}$, tolerance $\epsilon>0$, maximum iterations $I_D$, and iteration index $t=0$.
    \Repeat
        \For{each RA $n_t=1,\ldots,N_t$}
            \State Compute the element gradient $\mathbf g_{n_t}^{(t)}$ via \eqref{grad_block_D}.
            \State Project $\mathbf g_{n_t}^{(t)}$ onto the tangent space via \eqref{proj_grad_D}.
            \State Obtain the oracle point $\mathbf s_{n_t}^{(t)}$ via \eqref{oracle_solution_D}.
            \State Compute the search direction $\boldsymbol\Delta_{n_t}^{(t)}$ via \eqref{FW_direction_D}.
        \EndFor
        \State Compute the optimality gap $\sigma_{\rm FW}^{(t)}$ via \eqref{FW_gap_D}.
        \If{$\sigma_{\rm FW}^{(t)}\leq \epsilon$}
            \State \textbf{break}
        \EndIf
        \State Find step size $\rho^{(t)}\in(0,1]$ via Armijo backtracking to satisfy \eqref{Armijo_D}.
        \For{each RA $n_t=1,\ldots,N_t$}
            \State Update $\mathbf p_{n_t}^{(t+1)}$ via \eqref{D_update}.
        \EndFor
        \State Update $t=t+1$.
    \Until{The objective value $\tilde{\mathcal G}(\mathbf D)$ converges or the maximum iteration number $I_D$ is reached.}\\
\textbf{Output:} $\mathbf D^\star$
\end{algorithmic}
\end{algorithm}

\begin{proposition}\label{Proposition:2}
The proposed Frank--Wolfe--based boresight update generates a non-decreasing sequence of transformed objective values $\{\tilde{\mathcal G}(\mathbf D^{(t)})\}$. Moreover, this sequence is convergent.
\end{proposition}

\begin{proof}
See Appendix B.
\end{proof}

The Frank--Wolfe-based algorithm for solving Problem \eqref{P_D} is summarized in Algorithm~\ref{alg:D_update}.

\subsection{Overall Algorithm}
\begin{algorithm}[!t]
    \caption{Proposed AO Algorithm for RA-Enabled Near-Field ISAC Design}
    \label{alg:final_algorithm}
\begin{algorithmic}[1]
    \State \textbf{Initialize:} $\mathbf u^{(0)}$, $\mathbf F^{(0)}$, $\mathbf W^{(0)}$, $\mathbf D^{(0)}$, $\boldsymbol\mu^{(0)}$, $\boldsymbol\xi^{c(0)}$, $\boldsymbol\xi^{s(0)}$, tolerance $\epsilon_{\rm AO}>0$, maximum iterations $I_{\rm AO}$, and iteration index $t=0$.
    \Repeat
        \State Obtain $\mathbf u^{(t)}$ via \eqref{eq:u_optimal}.
        \State Obtain $\mathbf W^{(t)}$ via \eqref{eq:w_lambda}.
        \State Update $\mathbf F^{(t)}$ by solving \eqref{F_optimize} with the proposed Riemannian conjugate gradient method.
        \State Update $\boldsymbol\mu^{(t)}$ via \eqref{muk} and \eqref{muK+1}.
        \State Update $\boldsymbol\xi^{c(t)}$ and $\boldsymbol\xi^{s(t)}$ via \eqref{xic} and \eqref{xis}, respectively.
        \State Update $\mathbf D^{(t)}$ by Algorithm~\ref{alg:D_update}.
        \State Update $t=t+1$.
    \Until{The objective value of problem \eqref{P_1} converges or the maximum iterations $I_{\rm AO}$ is reached.}\\
\textbf{Output:} $\mathbf u^\star$, $\mathbf F^\star$, $\mathbf W^\star$, $\mathbf D^\star$.
\end{algorithmic}
\end{algorithm}
The detailed overall AO algorithm for solving the problem \eqref{P_1} is summarized in $\textbf{Algorithm}$ \ref{alg:final_algorithm}.
In each outer iteration, the receive beamformer $\mathbf u$, digital beamformer $\mathbf W$, analog beamformer $\mathbf F$, auxiliary variables $\{\boldsymbol\mu,\boldsymbol\xi^c,\boldsymbol\xi^s\}$, and RA boresight directions $\mathbf D$ are updated alternately with the other variables fixed. Each subproblem is solved either optimally or by a monotonic update rule, and thus the transformed objective value $\tilde{\mathcal G}$ is non-decreasing over the AO iterations. Since the transmit power is finite and the feasible sets of $\mathbf u$, $\mathbf F$, and $\mathbf D$ are compact, the objective sequence is upper bounded. Therefore, the objective sequence generated by Algorithm~\ref{alg:final_algorithm} is convergent.

For Algorithm~\ref{alg:final_algorithm}, let $I_{\rm AO}$ be the outer AO iterations, $I_F$ the number of Riemannian conjugate gradient iterations, and $I_D$ the Frank--Wolfe-based steps. The receive beamforming block is dominated by the matrix inversion in \eqref{eq:u_optimal}, yielding $\mathcal O(N_r^3)$ complexity. The digital beamforming block mainly requires solving a $B$-dimensional linear system in \eqref{eq:w_lambda}, resulting in $\mathcal O(B^3)$ complexity. For the analog beamforming block, each Riemannian conjugate gradient iteration is dominated by the matrix-vector product associated with $\tilde{\boldsymbol\Xi}$, yielding $\mathcal O(N_t^2)$ per iteration and thus $\mathcal O(I_FN_t^2)$ in total. The auxiliary-variable updates are in closed form and incur lower-order complexity.
For the RA boresight block, the gradients with respect to all element-wise
RA boresights require approximately
$\mathcal O(N_t(K+1)(K+C))$ per Frank--Wolfe-based step. The tangent
projection, spherical-cap oracle, search-direction construction, and
boresight update are performed over $N_t$ RA boresight vectors and incur
only lower-order $\mathcal O(N_t)$ complexity. Hence, the RA boresight
block costs $\mathcal O(I_DN_t(K+1)(K+C))$ per outer iteration.
Therefore, the overall complexity of Algorithm~\ref{alg:final_algorithm} is approximately given by
$\mathcal O(
    I_{\rm AO}
    (
    N_r^3
    +
    B^3
    +
    I_FN_t^2
    +
    I_DN_t(K+1)(K+C)
    )
    )$.

\section{Numerical Results}
In this section, numerical results are provided to evaluate the performance of the proposed RA-enabled near-field beamforming-ISAC design.

\subsection{Simulation Setup}
In the simulations, the carrier frequency is set to $f_c=30$~GHz. The BS
is equipped with an $8\times 8$ transmit UPA with $N_t=64$ RAs and a
$4\times 4$ receive UPA with $N_r=16$ FPAs. The side aperture of each
UPA is set to $D_{\rm side}=50\lambda$, and thus the maximum array
aperture is $D_{\rm ap}=\sqrt{2}D_{\rm side}$. The corresponding
Rayleigh distance is
$d_{\rm Ray}=2D_{\rm ap}^2/\lambda=100~\mathrm{m}$.
The sub-connected hybrid transmitter employs $B=8$ RF chains, each
connected to $M=8$ antennas. Each transmit RA independently adjusts its
boresight direction subject to the maximum rotation angle constraint.
Unless otherwise specified, the transmit power is set to $P=-10$~dBm.
The communication-sensing weights are set to
$\varpi_c=\varpi_s=0.5$, and the user weights are set to
$\alpha_k=1/K$. For the tradeoff evaluation, $\varpi_c$ is varied from
$0$ to $1$ with $\varpi_s=1-\varpi_c$.
The receiver noise power is computed as
$\sigma^2=-174+10\log_{10}(B_{\rm w})+\mathrm{NF}$ dBm, where the
bandwidth is $B_{\rm w}=100$~MHz and the noise figure is
$\mathrm{NF}=7$~dB, yielding $\sigma^2=-87$~dBm. The RCS of the target and clutters are set to -5 dbsm and -7 dbsm. The RA directivity
factor is $p=2$, and the maximum rotation angle is
$\theta_{\max}=60^\circ$.
The BS serves $K=4$ users and senses one target in the
presence of $C=5$ clutters. The user and clutter distances are uniformly
generated from $[15,30]$~m, their azimuth angles from
$[-80^\circ,80^\circ]$, and their elevation angles from
$[0^\circ,20^\circ]$. The target is fixed at
$(15~\mathrm{m},60^\circ,0^\circ)$. Since the maximum considered range is
$30$~m, all users, the target, and clutters are located in the near-field
region. Each user channel contains one LoS path and $L=8$ NLoS paths.
The maximum numbers of AO and RA boresight-update iterations are $50$
and $100$, respectively, and the AO convergence tolerance is $10^{-4}$.
All results are averaged over $200$ Monte Carlo realizations.

For comparison, the following schemes are considered:
\begin{itemize}
    \item \textbf{FPA:} Fixed-position array with isotropic antenna elements and sub-connected hybrid beamforming.

    \item \textbf{Fixed-RA:} Fixed-position RA array with sub-connected hybrid beamforming, where all RAs are oriented toward the array broadside.

    \item \textbf{Element RA:} Element-wise RA array with
    hybrid beamforming, where each transmit RA has an independently optimized
    boresight direction.

    \item \textbf{Fully-digital FPA:} Fixed-position array with fully-digital beamforming.
\end{itemize}
\subsection{Effect of antenna rotation on sensing performance}
\begin{figure*}[!t]
    \centering

    \begin{subfigure}{0.24\linewidth}
        \centering
        \includegraphics[width=\linewidth]{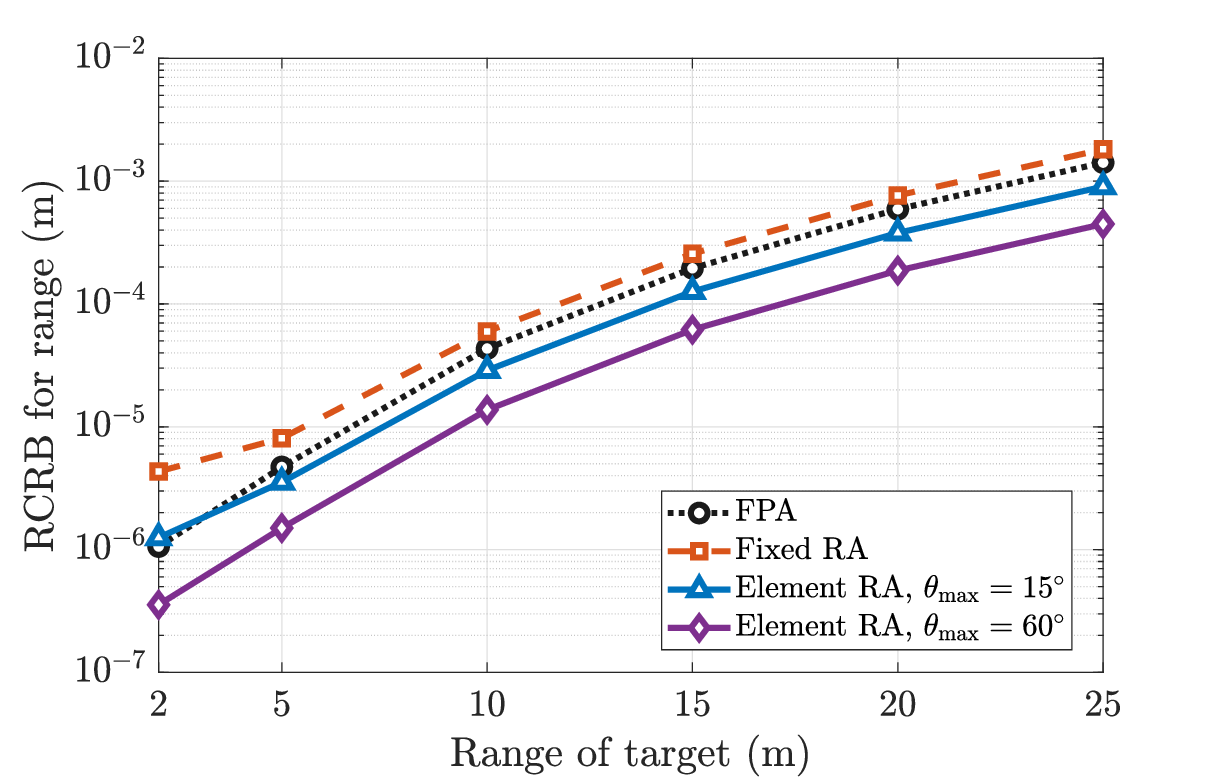}
        \caption{\justifying RCRB for range estimation versus target range.}
        \label{fig:CRB1}
    \end{subfigure}
    \hfill
    \begin{subfigure}{0.24\linewidth}
        \centering
        \includegraphics[width=\linewidth]{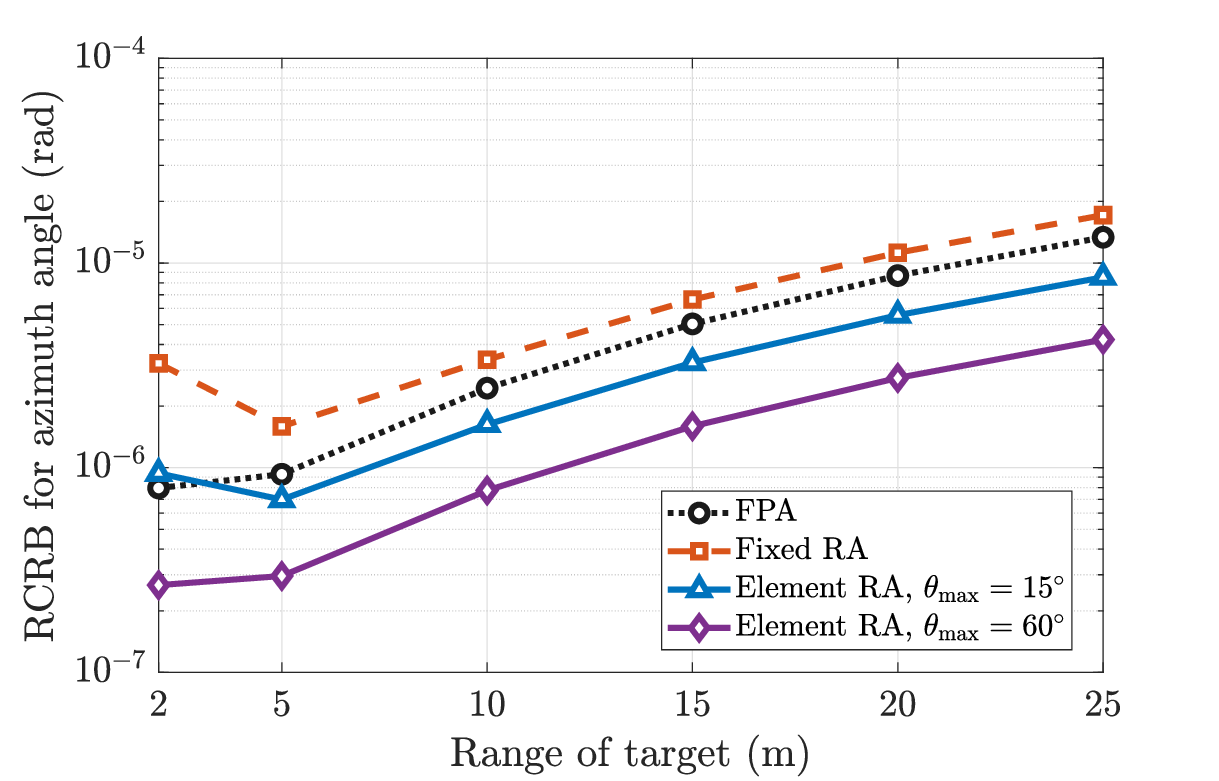}
        \caption{\justifying RCRB for angle estimation versus target range.}
        \label{fig:CRB2}
    \end{subfigure}
    \hfill
    \begin{subfigure}{0.24\linewidth}
        \centering
        \includegraphics[width=\linewidth]{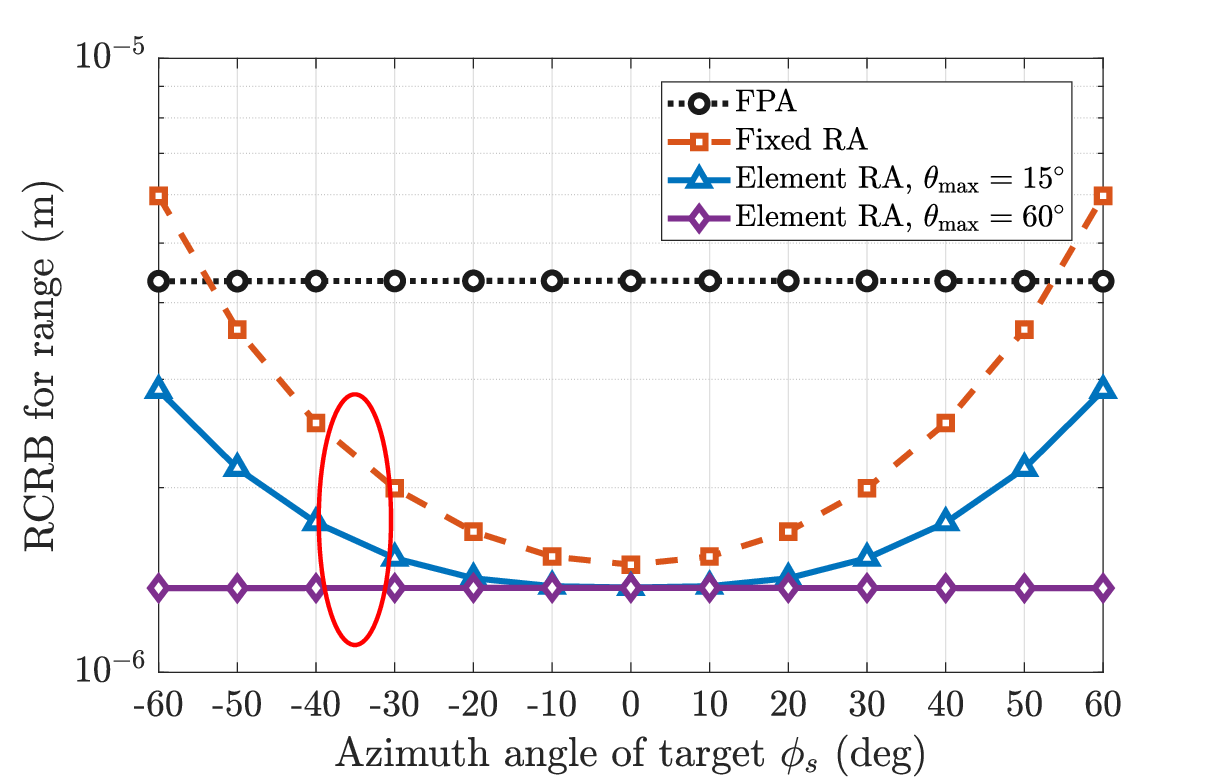}
        \caption{\justifying RCRB for range estimation versus target angle.}
        \label{fig:CRB3}
    \end{subfigure}
    \hfill
    \begin{subfigure}{0.24\linewidth}
        \centering
        \includegraphics[width=\linewidth]{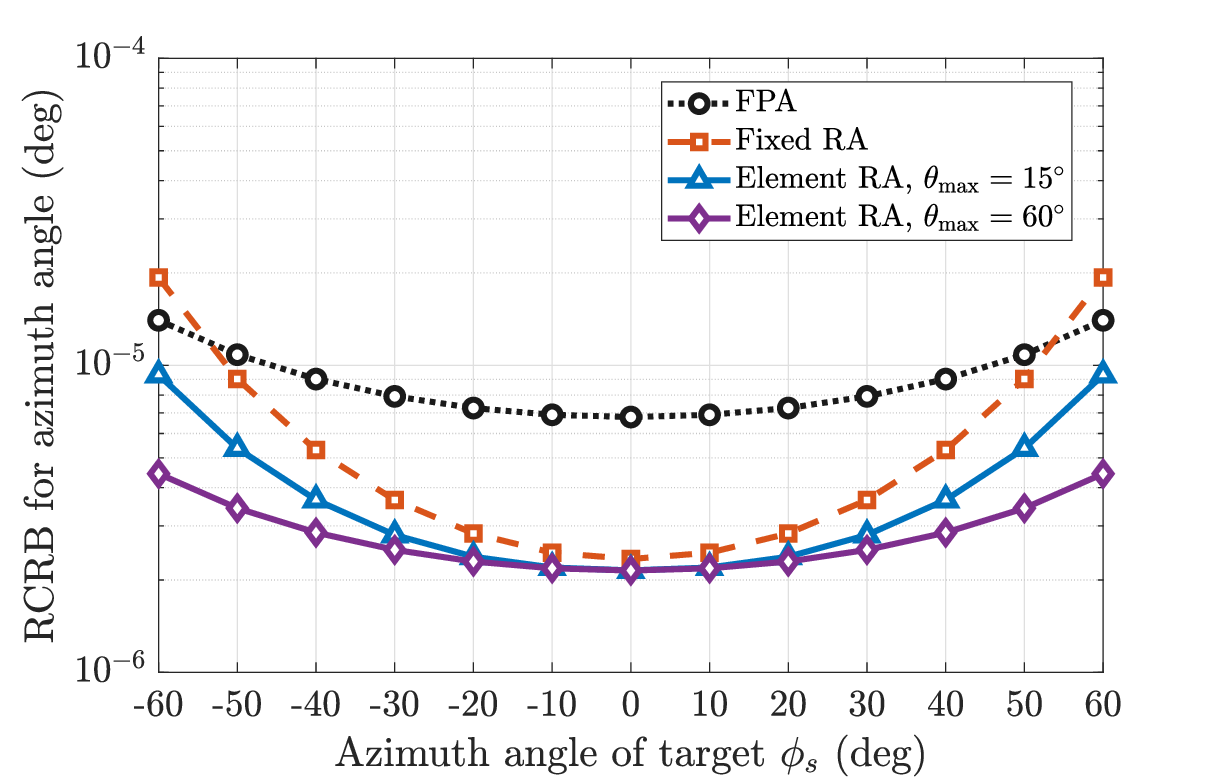}
        \caption{\justifying RCRB for angle estimation versus target angle.}
        \label{fig:CRB4}
    \end{subfigure}
    \captionsetup{font={small}}
    \caption{\justifying Effect of antenna rotation on near-field sensing accuracy.}
    \label{fig:CRB_all}
\end{figure*}

To further reveal the impact of antenna rotation on near-field sensing,
we evaluate the RCRB of the target location parameters. The unknown
complex target response is treated as a nuisance parameter, and the
detailed RCRB derivation is provided in Appendix~\ref{app:RCRB}. To
isolate the effect of RA boresight control, an equal-power orthogonal
probing waveform is adopted for all schemes.
For each considered target location, the RA boresights are configured according to the corresponding target direction under the rotation constraint.

Figs.~\ref{fig:CRB1} and \ref{fig:CRB2} show the range and azimuth RCRBs
versus the target range, respectively, where the target azimuth is fixed
at $45^\circ$. The RCRBs increase with the target range due to the
stronger round-trip propagation loss and the reduced echo sensitivity.
For Element RA with $\theta_{\max}=15^\circ$, the allowable rotation
range is insufficient to cover the $45^\circ$ off-broadside target, and
thus the sensing gain is limited. In contrast, when
$\theta_{\max}=60^\circ$, the RA boresights can be aligned with the
target direction, leading to much lower RCRBs than FPA and Fixed RA.
This confirms that the sensing gain of RA depends critically on whether
the rotation range can cover the target direction.

Figs.~\ref{fig:CRB3} and \ref{fig:CRB4} show the range and azimuth RCRBs
versus the target azimuth angle, respectively, with the target range
fixed at $10$~m. Fixed RA performs well near broadside but degrades
rapidly for off-broadside targets because its boresights are fixed. In
contrast, Element RA maintains lower RCRBs by adapting its boresights to
the target direction. When $\theta_{\max}=15^\circ$, the RCRBs increase
once the target moves outside the feasible rotation range, since the RAs
cannot fully align with the target. With $\theta_{\max}=60^\circ$, the
considered angular region can be covered, yielding consistently low
RCRBs. These results verify that element-wise rotation improves
near-field sensing accuracy and enhances robustness to off-broadside
targets when the rotation range is sufficiently large.
\subsection{Simulation Results}
\begin{figure}[!t]
    \centering
    \includegraphics[width=\linewidth]{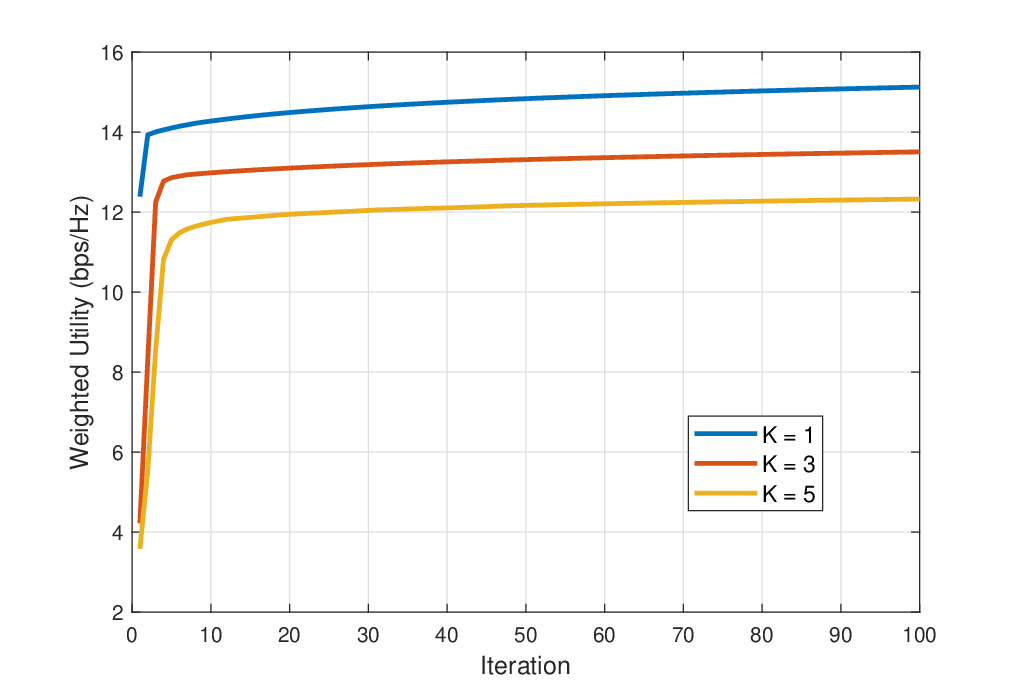}
    \captionsetup{font={small}}
    \caption{\justifying Convergence behavior of Algorithm~2.}
    \label{fig:convergence}
\end{figure}
Fig.~\ref{fig:convergence} shows the convergence behavior of
Algorithm~2 under different numbers of users. The weighted utility
increases rapidly in the first few iterations and becomes nearly stable
after about $10$ iterations, which verifies the convergence efficiency of
the proposed AO algorithm. As $K$ increases, the final utility decreases
because more communication streams have to share the same transmit power
and RF-chain resources. Meanwhile, stronger multi-user interference
makes the joint communication-sensing design more challenging.

\begin{figure}[t]
    \centering
    \includegraphics[width=\linewidth]{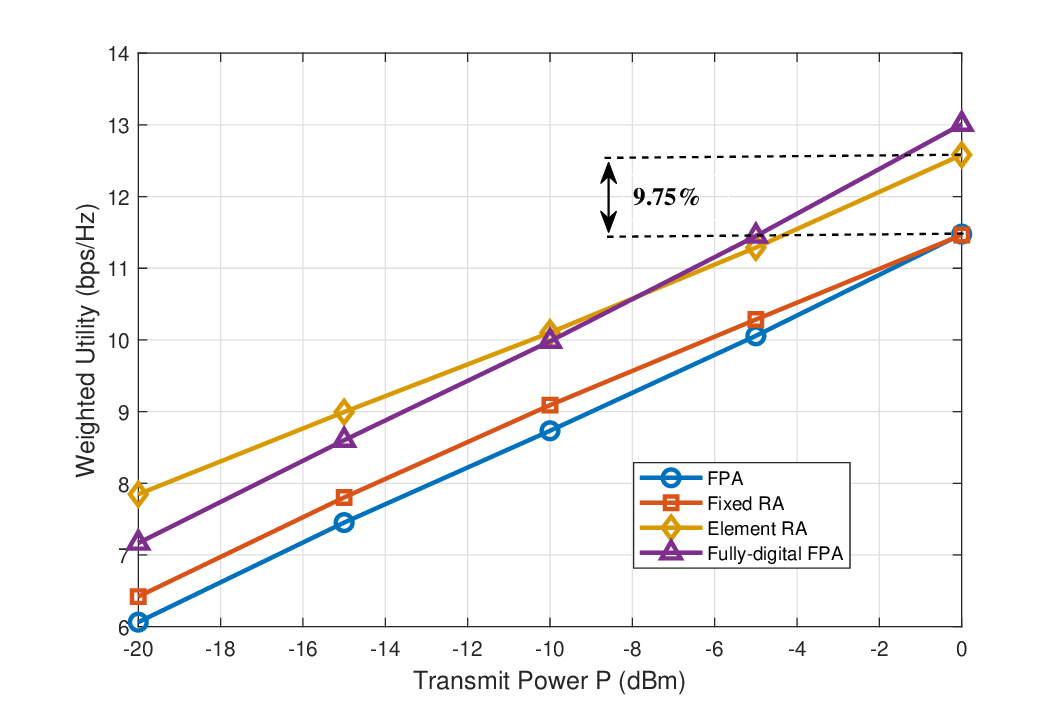}
    \captionsetup{font={small}}
    \caption{\justifying Weighted utility versus transmit power.}
    \label{fig:P}
\end{figure}
Fig.~\ref{fig:P} compares the weighted communication-sensing utility
versus the transmit power. All schemes improve with $P$ due to the
enhanced communication SINR and sensing SCNR. In the low-power regime,
the proposed Element RA achieves comparable or higher utility than
Fully-digital FPA, since the directional gain brought by antenna rotation
is particularly effective when the system is power-limited. As $P$
increases, Fully-digital FPA becomes slightly superior because the system
is less noise-limited and its larger spatial degrees of freedom are more
effective for interference and clutter suppression. Nevertheless, Element
RA consistently outperforms FPA and Fixed RA, showing the benefit of
adaptive element-wise boresight control.

\begin{figure}[t]
    \centering
    \includegraphics[width=\linewidth]{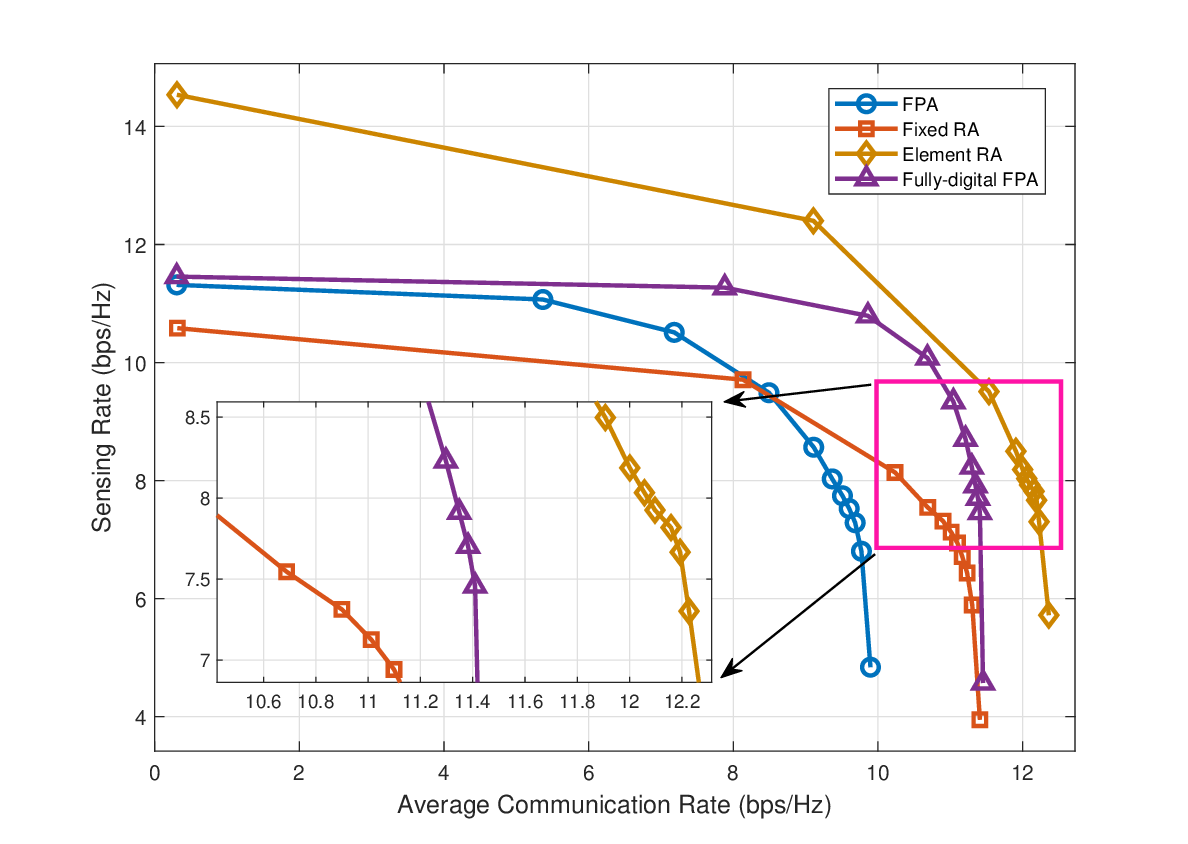}
    \captionsetup{font={small}}
    \caption{\justifying Tradeoff between sensing and communication performance.}
    \label{fig:tradeoff}
\end{figure}
Fig.~\ref{fig:tradeoff} illustrates the communication-sensing tradeoff by
varying the weighting factors. Increasing the communication rate
generally reduces the sensing rate, which reflects the inherent resource
competition between communication and sensing. The proposed Element RA
moves the tradeoff curve toward the upper-right region, indicating that
antenna rotation improves the achievable ISAC Pareto frontier rather
than only one individual metric. Compared with Fully-digital FPA,
Element RA achieves a better tradeoff because the additional
orientation-domain gain enhances both user links and target illumination.

\begin{figure}[t]
    \centering
    \includegraphics[width=\linewidth]{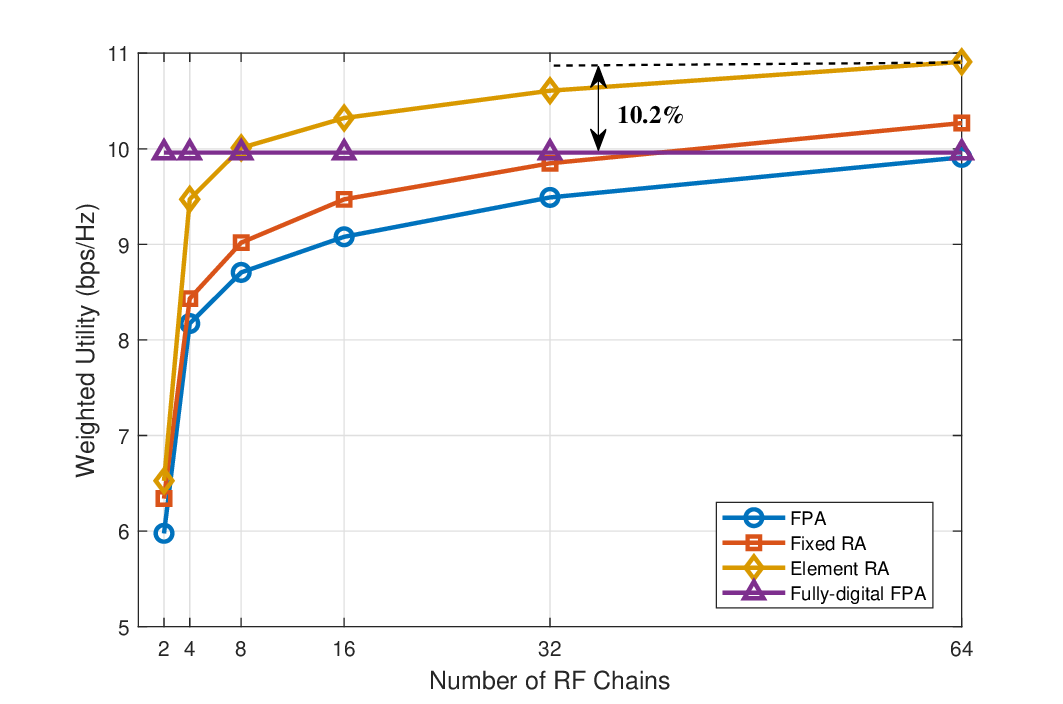}
    \captionsetup{font={small}}
    \caption{\justifying Weighted utility versus number of RF chains.}
    \label{fig:M}
\end{figure}
Fig.~\ref{fig:M} shows the weighted communication-sensing utility versus
the number of RF chains. The utility of all hybrid schemes increases with
$B$, since more RF chains provide higher digital precoding flexibility.
The improvement is more significant when $B$ is small and gradually
saturates as $B$ increases. Fully-digital FPA is independent of $B$ and
serves as a reference. The proposed Element RA approaches this reference
with only a small number of RF chains and surpasses it when $B$ becomes
large. This demonstrates that antenna rotation can effectively alleviate
the RF-chain limitation of sub-connected hybrid beamforming.

\begin{figure}[t]
    \centering
    \includegraphics[width=\linewidth]{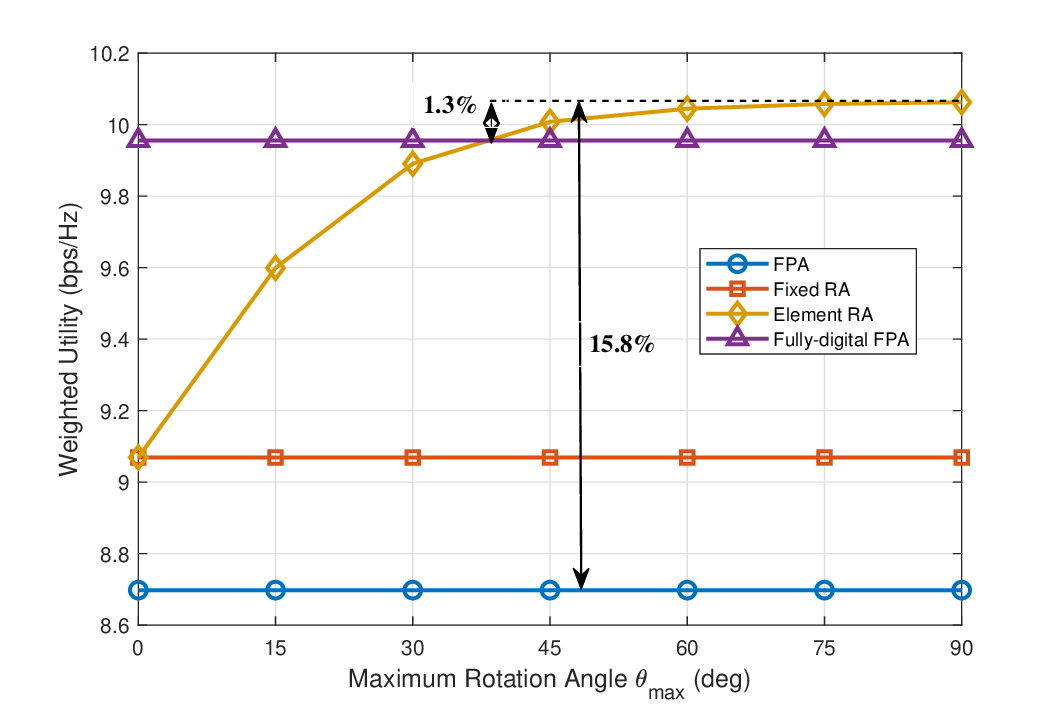}
    \captionsetup{font={small}}
    \caption{\justifying Weighted utility versus allowable maximum rotation angle.}
    \label{fig:ThetaMax}
\end{figure}
Fig.~\ref{fig:ThetaMax} shows the weighted utility versus the maximum
rotation angle $\theta_{\max}$. The FPA, Fixed RA, and Fully-digital FPA
schemes remain unchanged since they are independent of
$\theta_{\max}$. When $\theta_{\max}=0^\circ$, Element RA reduces to
Fixed RA. As $\theta_{\max}$ increases, Element RA gains more freedom to
align its boresights with off-broadside users and the target, resulting
in a higher utility. The performance gradually saturates when the
rotation range is sufficiently large, which indicates that a moderate
rotation capability is enough to capture most of the orientation-domain
gain. This result provides a useful guideline for practical RA hardware
design.

\begin{figure}[t]
    \centering
    \includegraphics[width=\linewidth]{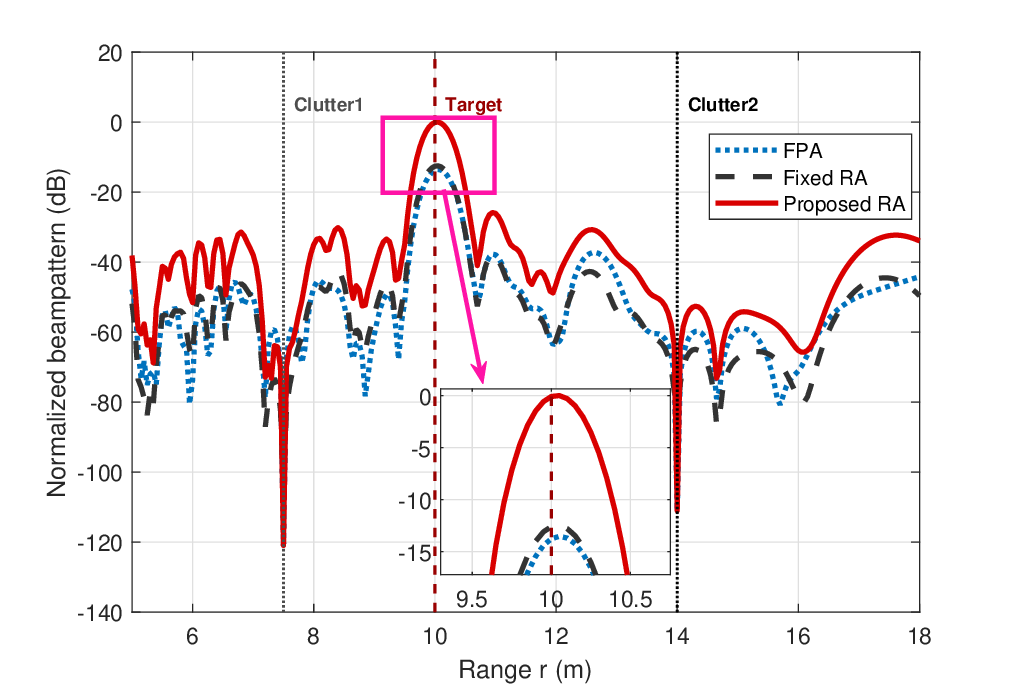}
    \captionsetup{font={small}}
    \caption{\justifying Beampattern comparison of different schemes.}
    \label{fig:beampattern}
\end{figure}
Fig.~\ref{fig:beampattern} shows the normalized beampattern along the
target direction, where the target and two clutters share the same
azimuth angle but are located at different ranges. The beampattern peaks
around the target range and exhibits deep nulls at the clutter ranges,
demonstrating the range-domain focusing capability of near-field
beamforming. This is fundamentally different from far-field beamforming,
where same-angle objects are difficult to separate. Compared with FPA
and Fixed RA, the proposed Element RA achieves stronger target
illumination and deeper same-angle clutter suppression, which verifies
the effectiveness of RA boresight optimization for near-field sensing.

\section{Conclusion}
In this paper, we studied an RA-enabled near-field ISAC
system with sub-connected hybrid beamforming. By incorporating
spherical-wave propagation and orientation-dependent antenna gains, we
developed a near-field channel model for multi-user communication and
target sensing in the presence of clutters. We formulated a weighted
communication-sensing utility maximization problem by jointly optimizing
the receive sensing beamformer, hybrid transmit beamformer, and RA
boresight directions. To solve the resulting non-convex problem, an
AO algorithm was proposed based on fractional programming, Riemannian
optimization, and a spherical-cap Frank-Wolfe-based method.
Numerical results demonstrated the convergence and effectiveness of the proposed algorithm.
 The RCRB and beampattern results further demonstrated that
antenna rotation improves sensing accuracy and 
clutter suppression in the near field. These results confirm the
potential of element-wise RA for enhancing the communication-sensing
tradeoff in near-field ISAC systems.

\appendices
\section{Proof of Proposition~\ref{Proposition:1}}
For given $\{\mathbf F,\mathbf W,\mathbf D\}$, define
\begin{equation}
    \mathbf R_{\rm i}
    =
    \sum_{c=1}^{C}
    \tilde{\mathbf H}_c\tilde{\mathbf H}_c^{\rm H}
    +
    \sigma_s^2\mathbf I_{N_r},
    \quad
    \tilde{\mathbf H}_c=\mathbf H_c\mathbf F\mathbf W .
\end{equation}
Since $\sigma_s^2>0$, we have $\mathbf R_{\rm i}\succ\mathbf 0$.
According to the point-like target channel in \eqref{target echo} and by defining
$\mathbf a_s^{\rm H}=\mathbf g(\mathbf q_s,\mathbf D)\odot
\mathbf b^{\rm H}(\mathbf q_s)$, the SCNR can be rewritten as
\begin{equation}
    {\rm SCNR}
    =
    \frac{
    |\beta_s|^2
    \left\|\mathbf a_s^{\rm H}\mathbf F\mathbf W\right\|_2^2
    \left|\mathbf u^{\rm H}\mathbf b_r(\mathbf q_s)\right|^2
    }{
    \mathbf u^{\rm H}\mathbf R_{\rm i}\mathbf u
    }.
\end{equation}
Since the scalar
$|\beta_s|^2\left\|\mathbf a_s^{\rm H}\mathbf F\mathbf W\right\|_2^2$
is independent of $\mathbf u$, problem \eqref{u_optimize} is equivalent to
\begin{equation}
    \max_{\|\mathbf u\|_2^2=1}
    \frac{
    \left|\mathbf u^{\rm H}\mathbf b_r(\mathbf q_s)\right|^2
    }{
    \mathbf u^{\rm H}\mathbf R_{\rm i}\mathbf u
    }.
\end{equation}
By the generalized Rayleigh quotient, the optimal direction is proportional to
$\mathbf R_{\rm i}^{-1}\mathbf b_r(\mathbf q_s)$. Normalizing it yields
\eqref{eq:u_optimal}, which completes the proof.

\section{Proof of Proposition 2}
For each RA $n_t$, the spherical-cap oracle is given by \eqref{oracle_D}.
Since the current boresight vector $\mathbf p_{n_t}^{(t)}$ is feasible, i.e., $\mathbf p_{n_t}^{(t)}\in\mathcal C_{\rm cap}$, the optimality of $\mathbf s_{n_t}^{(t)}$ implies
\begin{equation}
\left(\bar{\mathbf g}_{n_t}^{(t)}\right)^T\mathbf s_{n_t}^{(t)}
\ge
\left(\bar{\mathbf g}_{n_t}^{(t)}\right)^T\mathbf p_{n_t}^{(t)} .
\end{equation}
Using the definition $\boldsymbol\Delta_{n_t}^{(t)}=\mathbf s_{n_t}^{(t)}-\mathbf p_{n_t}^{(t)}$, we obtain
\begin{equation}
\left(\bar{\mathbf g}_{n_t}^{(t)}\right)^T
\boldsymbol\Delta_{n_t}^{(t)}
\ge 0 .
\end{equation}
Summing over all subarrays yields
\begin{equation}
\sigma_{\rm FW}^{(t)}
=
\sum_{{n_t}=1}^{{N_t}}
\left(\bar{\mathbf g}_{n_t}^{(t)}\right)^T
\boldsymbol\Delta_{n_t}^{(t)}
\ge 0 .
\end{equation}

For each accepted step size $\rho^{(t)}$, the Armijo sufficient-increase condition gives
\begin{equation}
\tilde{\mathcal G}(\mathbf D^{(t+1)})
\ge
\tilde{\mathcal G}(\mathbf D^{(t)})
+
c_A\rho^{(t)}\sigma_{\rm FW}^{(t)} .
\end{equation}
Since $c_A>0$, $\rho^{(t)}>0$, and $\sigma_{\rm FW}^{(t)}\ge0$, it follows that
\begin{equation}
\tilde{\mathcal G}(\mathbf D^{(t+1)})
\ge
\tilde{\mathcal G}(\mathbf D^{(t)}),
\end{equation}
which proves that the transformed objective value is non-decreasing.

Moreover, each $\mathbf p_{n_t}$ is constrained in the spherical cap $\mathcal C_{\rm cap}$, which is closed and bounded. Therefore, the feasible set $\mathcal C$ is compact. Since the transmit power is finite and the transformed objective $\tilde{\mathcal G}(\mathbf D)$ is continuous over the compact feasible set, $\tilde{\mathcal G}(\mathbf D)$ is upper bounded. Hence, the non-decreasing objective sequence $\{\tilde{\mathcal G}(\mathbf D^{(t)})\}$ is bounded from above and thus convergent. This completes the proof.
\section{RCRB Derivation for Near-Field Sensing}
\label{app:RCRB}

Let $\mathbf X=[\mathbf x(1),\ldots,\mathbf x(T)]\in
\mathbb C^{N_t\times T}$ denote the known probing matrix over $T$
sensing symbols. To decouple the sensing accuracy from a particular
beamformer realization and focus on the antenna-rotation effect, an
equal-power orthogonal probing matrix is adopted, satisfying
\begin{equation}
\frac{1}{T}\mathbf X\mathbf X^H=\frac{P}{N_t}\mathbf I_{N_t}.
\end{equation}

The received echo before receive combining is
\begin{equation}
\mathbf Y
=
\beta_s\mathbf A_s(\boldsymbol{\eta})\mathbf X
+
\sum_{c=1}^{C}\beta_c\mathbf A_c\mathbf X
+
\mathbf N,
\end{equation}
where $\boldsymbol{\eta}=[r_s,\vartheta_s,\phi_s]^T$ is the target
location parameter vector, $\beta_s$ is the unknown complex target
response, and $\mathbf A_s(\boldsymbol{\eta})$ denotes the target
response. The matrix $\mathbf A_c$ is defined
similarly for the $c$-th clutter. By vectorizing $\mathbf Y$, the observation can be written as
\begin{equation}
\begin{aligned}
\mathbf y
&=
\beta_s\mathbf b(\boldsymbol{\eta})+\sum_{c=1}^{C}\beta_c\mathbf b_c+\mathbf n.
\end{aligned}
\end{equation}
where $\mathbf b(\boldsymbol{\eta})=\operatorname{vec}\big(\mathbf A_s(\boldsymbol{\eta})\mathbf X\big)$ and $\mathbf b_c= \operatorname{vec}(\mathbf A_c\mathbf X)$.
The clutter coefficients are modeled as independent zero-mean complex
Gaussian variables with
$\mathbb E\{|\beta_c|^2\}=\sigma_{\beta,c}^2$. Hence, the clutter-plus-noise
term is an effective Gaussian disturbance with covariance
\begin{equation}
\mathbf R_v
=
\sum_{c=1}^{C}\sigma_{\beta,c}^2\mathbf b_c\mathbf b_c^H
+
\sigma_s^2\mathbf I .
\end{equation}

The target response matrix is expressed as
$\mathbf A_s(\boldsymbol{\eta})
=\mathbf a_r(\boldsymbol{\eta})\mathbf a_t^H(\boldsymbol{\eta})$.
Let $\mathbf q_s=\mathbf q(r_s,\vartheta_s,\phi_s)$. For the $n$-th
antenna of array $\ell\in\{t,r\}$, define
$d_{\ell,n}=\|\mathbf q_s-\mathbf t_{\ell,n}\|_2$ and
$\mathbf s_{\ell,n}=(\mathbf q_s-\mathbf t_{\ell,n})/d_{\ell,n}$.
For $\eta_i\in\{r_s,\vartheta_s,\phi_s\}$, their derivatives are
\begin{equation}
\begin{aligned}
\frac{\partial d_{\ell,n}}{\partial \eta_i}
&=
\mathbf s_{\ell,n}^T
\frac{\partial \mathbf q_s}{\partial \eta_i},\\
\frac{\partial \mathbf s_{\ell,n}}{\partial \eta_i}
&=
\frac{\mathbf I_3-\mathbf s_{\ell,n}\mathbf s_{\ell,n}^T}{d_{\ell,n}}
\frac{\partial \mathbf q_s}{\partial \eta_i}.
\end{aligned}
\end{equation}

The transmit and receive near-field responses are
\begin{equation}
[\mathbf a_t^H]_n
=
\gamma_n e^{-j\frac{2\pi}{\lambda}(d_{t,n}-r_s)},\quad
[\mathbf a_r]_m
=
e^{-j\frac{2\pi}{\lambda}(d_{r,m}-r_s)}.
\end{equation}
For FPA, $\gamma_n=1$. For RA,
$\gamma_n=\sqrt{G_0}(\mathbf p_n^T\mathbf s_{t,n})^p$ in the active
region $\mathbf p_n^T\mathbf s_{t,n}>0$. Accordingly, the response
derivatives are given by
\begin{equation}
\begin{aligned}
\frac{\partial [\mathbf a_t^H]_n}{\partial \eta_i}
&=
[\mathbf a_t^H]_n
\left[
\chi_{n,i}
-
j\frac{2\pi}{\lambda}
\left(
\frac{\partial d_{t,n}}{\partial \eta_i}-\iota_i
\right)
\right],\\
\frac{\partial [\mathbf a_r]_m}{\partial \eta_i}
&=
[\mathbf a_r]_m
\left[
-
j\frac{2\pi}{\lambda}
\left(
\frac{\partial d_{r,m}}{\partial \eta_i}-\iota_i
\right)
\right],
\end{aligned}
\end{equation}
where $\iota_i=1$ if $\eta_i=r_s$ and $\iota_i=0$ otherwise. For RA,
$\chi_{n,i}=p(\mathbf p_n^T\partial\mathbf s_{t,n}/\partial\eta_i)/
(\mathbf p_n^T\mathbf s_{t,n})$, while $\chi_{n,i}=0$ for FPA.
Accordingly,
$
\frac{\partial \mathbf b}{\partial \eta_i}
=
\operatorname{vec}\left[
\left(
\frac{\partial \mathbf a_r}{\partial \eta_i}\mathbf a_t^H
+
\mathbf a_r
\frac{\partial \mathbf a_t^H}{\partial \eta_i}
\right)\mathbf X
\right]
$. 

Define $
\dot{\mathbf B}
=
\left[
\frac{\partial \mathbf b}{\partial r_s},
\frac{\partial \mathbf b}{\partial \vartheta_s},
\frac{\partial \mathbf b}{\partial \phi_s}
\right].$ Since $\beta_s$ is unknown, it is treated as a nuisance parameter and
eliminated via the Schur complement. The equivalent FIM of
$\boldsymbol{\eta}$ is
\begin{equation}
\mathbf J_{\boldsymbol{\eta}}
=
2|\beta_s|^2
\Re\left\{
\dot{\mathbf B}^H
\mathbf P_{\mathbf b,\mathbf R_v}^{\perp}
\dot{\mathbf B}
\right\},
\end{equation}
where
$\mathbf P_{\mathbf b,\mathbf R_v}^{\perp}
=
\mathbf R_v^{-1}
-
\mathbf R_v^{-1}\mathbf b
\left(
\mathbf b^H\mathbf R_v^{-1}\mathbf b
\right)^{-1}
\mathbf b^H\mathbf R_v^{-1}$.
For the white effective disturbance case
$\mathbf R_v=\tilde{\sigma}_s^2\mathbf I$, it can be reduced to
\begin{equation}
\mathbf J_{\boldsymbol{\eta}}
=
\frac{2|\beta_s|^2}{\tilde{\sigma}_s^2}
\Re\left\{
\dot{\mathbf B}^H
\mathbf P_{\mathbf b}^{\perp}
\dot{\mathbf B}
\right\},
\quad
\mathbf P_{\mathbf b}^{\perp}
=
\mathbf I-\mathbf b(\mathbf b^H\mathbf b)^{-1}\mathbf b^H .
\end{equation}

Therefore, the RCRBs of interest are
\begin{align}
\mathrm{RCRB}_{r}
&=
\sqrt{[(\mathbf J_{\boldsymbol{\eta}})^{-1}]_{1,1}},\\
\mathrm{RCRB}_{\vartheta}
&=
\sqrt{[(\mathbf J_{\boldsymbol{\eta}})^{-1}]_{2,2}},\\
\mathrm{RCRB}_{\phi}
&=
\sqrt{[(\mathbf J_{\boldsymbol{\eta}})^{-1}]_{3,3}}.
\end{align}

\end{document}

\end{document}